\documentclass[11pt,a4paper]{article}

\usepackage{amssymb}
\usepackage{amsmath}
\usepackage{amsthm}
\usepackage{epsfig}
\usepackage{euler}
\usepackage[hang,stable]{footmisc}

\theoremstyle{plain}
\newtheorem{theorem}{Theorem}
\newtheorem{corollary}[theorem]{Corollary}
\newtheorem{proposition}[theorem]{Proposition}
\newtheorem{lemma}[theorem]{Lemma}
\theoremstyle{definition}
\newtheorem{definition}{Definition}

\newcommand{\reals}{\mathbb{R}}    
\newcommand{\complex}{\mathbb{C}}  
\newcommand{\rationals}{\mathbb{Q}}  
\newcommand{\field}{\mathbb{F}} 
\newcommand{\Mink}[1]{\mathbb{M}^{#1}}
\newcommand{\Perm}{\mathrm{Perm}} 
\newcommand{\End}{\mbox{End}} 
\newcommand{\Frames}[1]{\mathcal{F}_{#1}}     
\newcommand{\Stab}{\mathrm{Stab}}  
\newcommand{\Span}{\mbox{span}} 
\newcommand{\Proj}{\Pi} 
\newcommand{\Identity}{\mathrm{id}}
        
\newcommand{\Aff}{\mathrm{Aff}}
\newcommand{\Orb}{\mathrm{Orb}}
\newcommand{\group}{\mathsf}

\newcommand{\Lor}{\mathsf{Lor}} 
\newcommand{\ILor}{\mathsf{ILor}}

\newcommand{\mat}{\mathbf}

\makeindex             

\begin{document}
\ifx\href\undefined\else\hypersetup{linktocpage=true}\fi
\tolerance 10000

\title{The Rich Structure of Minkowski Space}
\author{Domenico Giulini                      \\
Max-Planck-Institute for Gravitational Physics\\
(Albert-Einstein-Institute)                   \\
Am M\"uhlenberg 1                             \\ 
D-14476 Golm/Potsdam, Germany                 \\
\texttt{domenico.giulini@aei.mpg.de}          }

\date{}
\maketitle

\begin{abstract}
\noindent
Minkowski Space is the simplest four-dimensional Lorentzian 
Manifold, being topologically trivial and globally flat, and 
hence the simplest model of spacetime---from a General-Relativistic 
point of view. But this does not mean that it is altogether 
structurally trivial. In fact, it has a very rich structure, 
parts of which will be spelled out in detail in this 
contribution, which is written for \emph{Minkowski Spacetime: A Hundred Years Later}, 
edited by Vesselin Petkov, to appear in 2008 in the Springer Series 
on Fundamental Theories of Physics, Springer Verlag, Berlin. 
\end{abstract}

\begin{small}
\setcounter{tocdepth}{3}
\tableofcontents
\end{small}
\newpage

\section{General Introduction}
\label{sec:GenIntro}
There are many routes to Minkowski space. But the most physical one 
still seems to me via the law of inertia. And even along these lines 
alternative approaches exist. Many papers were published in physics and 
mathematics journals over the last 100 years in which incremental 
progress was reported as regards the minimal set of hypotheses 
from which the structure of Minkowski space could be deduced. 
One could imagine a Hesse-diagram-like picture in which all these 
contributions (being the nodes) together with their logical
dependencies (being the directed links) were depicted. It would 
look surprisingly complex. 

From a General-Relativistic point of view, Minkowski space just 
models an \emph{empty} spacetime, that is, a spacetime devoid of any 
material content. It is worth keeping in mind, that this was not 
Minkowski's view. Close to the beginning of \emph{Raum und Zeit} he 
stated:\footnote{German original:
``Um nirgends eine g\"ahnende Leere zu lassen, wollen  wir uns 
vorstellen, da{\ss} allerorten und zu jeder Zeit etwas Wahrnehmbares 
vorhanden ist''. (\cite{Minkowski:1909}, p.\,2)}
\begin{quote}
\emph{
In order to not leave a yawning void, we wish to imagine that at 
every place and at every time something perceivable exists.}
\end{quote}
This already touches upon a critical point. Our modern theoretical 
view of spacetime is much inspired by the typical hierarchical 
thinking of mathematics of the late 19th and first half of the 
20th century, in which the \emph{set} comes first, and then we 
add various structures on it. We first think of spacetime as a 
set and then structure it according to various physical inputs. 
But what are the elements of this set? Recall how Georg Cantor, 
in  his first article on transfinite set-theory, defined a set:%
\footnote{German original: ``Unter einer `Menge' 
verstehen wir jede  Zusammenfassung $M$ von bestimmten 
wohlunterschiedenen Objecten $m$  unserer Anschauung oder unseres 
Denkens (welche die `Elemente' von $M$ genannt werden) zu einem 
Ganzen.'' (\cite{CantorMengenlehre1:1895}, p.\,481)} 
\begin{quote}
\emph{
By a `set' we understand any gathering-together $M$ of determined well-distinguished 
objects $m$ of our intuition or of our thinking (which are called the `elements' of $M$)  
into a whole.}
\end{quote}
Do we think of spacetime points as ``determined well-distinguished 
objects of our intuition or of our thinking''? I think Minkowski
felt a need to do so, as his statement quoted above indicates, 
and also saw the problematic side of it: If we mentally individuate 
the points (elements) of spacetime, we---as physicists---have no 
other means to do so than to fill up spacetime with actual matter, 
hoping that this could be done in such a diluted fashion that 
this matter will not dynamically affect the processes that we 
are going to describe. In other words: The whole concept of 
a rigid background spacetime is, from its very beginning, based 
on an assumption of---at best---approximate validity. It is 
important to realise that this does not necessarily refer to 
General Relativity: Even if the need to incorporate gravity by 
a variable and matter-dependent spacetime geometry did not exist
would the concept of a rigid background spacetime be of approximate
nature, \emph{provided we think of spacetime points as individuated by actual 
physical events}. 

It is true that modern set theory regards Cantor's original 
definition as too na\"{\i}ve, and that for good reasons. 
It allows too many ``gatherings-together'' with self-contradictory 
properties, as exemplified by the infamous \emph{antinomies} of 
classical set theory. 
Also, modern set theory deliberately stands back from any 
characterisation of elements in order to not confuse the axioms 
themselves with their possible \emph{interpretations}.\footnote{
This urge for a clean distinction between the axioms and their possible 
interpretations is contained in the famous and amusing dictum, attributed 
to David Hilbert by his student Otto Blumenthal: ``One must always be 
able to say 'tables', `chairs', and `beer mugs' instead of 'points, 
`lines', and `planes''. (German original: ``Man mu\ss\ jederzeit an Stelle 
von 'Punkten', `Geraden' und `Ebenen' 'Tische', `St\"uhle' und `Bierseidel' 
sagen k\"onnen.'')} However, applications to physics require 
\emph{interpreted} axioms, where it remains true that elements of sets 
are thought of as definite as in Cantors original definition.

Modern textbooks on Special Relativity have little to say about 
this, though an increasing unease seems to raise its voice from 
certain directions in the philosophy-of-science community; 
see, e.g., \cite{Brown.Pooley:2006}\cite{Brown:PhysicalRelativity}. 
Physicists sometimes tend to address points of spacetime as 
\emph{potential events}, but that always seemed to me like  
poetry\footnote{``And as imagination bodies forth The forms of things unknown, 
the poet's pen Turns them to shapes, and gives to airy nothing 
A local habitation and a name.'' (A Midsummer Night's Dream, 
Theseus at V,i)}, begging the question how a mere potentiality 
is actually used for individuation. To me the right attitude seems 
to admit that the operational justification of the notion of spacetime 
events is only approximately possible, but nevertheless allow it 
as primitive element of theorising. The only thing to keep in 
mind is to not take mathematical rigour for ultimate physical 
validity. The purpose of mathematical rigour is rather to establish 
the tightest possible bonds between basic assumptions (axioms) 
and decidable consequences. Only then can we---in principle---learn 
anything through falsification. 

The last remark opens another general issue, which is implicit  
in much of theoretical research, namely how to balance between 
attempted rigour in drawing consequences and attempted closeness 
to reality when formulating once starting platform (at the 
expense of rigour when drawing consequences). As the mathematical 
physicists Glance \& Wightman once formulated it in a different 
context (that of superselection rules in Quantum Mechanics): 
\begin{quote}
\emph{
The theoretical results currently available fall into 
two categories: rigorous results on approximate models
and approximate results in realistic models.} 
(\cite{Wightman.Glance:1989}, p.\,204)
\end{quote}
To me this seems to be the generic situation in theoretical physics.
In that respect, Minkowski space is certainly an approximate model, 
but to a very good approximation indeed: as global model of spacetime 
if gravity plays no dynamical r\^ole, and as local model of spacetime 
in far more general situations. This justifies looking at some of 
its rich mathematical structures in detail. Some mathematical 
background material is provided in the Appendices. 

\section{Minkowski space and its partial automorphisms}
\label{sec:MinkowskiAutomorphisms}
\subsection{Outline of general strategy}
\label{sec:OutlineStrategy}
Consider first the general situation where one is given a set $S$. 
Without any further structure being specified, the \emph{automorphisms} 
group of $S$ would be the group of bijections of $S$, i.e. maps 
$f:S\rightarrow S$ which are injective (into) and surjective (onto). 
It is called $\Perm(S)$, where `Perm' stands for `permutations'. 
Now endow $S$ with some structure $\Delta$; for example, it could 
be an equivalence relation on $S$, that is, a partition of $S$ into 
an exhaustive set of mutually disjoint subsets 
(cf. Sect.\,\ref{sec:GroupActions}). The automorphism group of 
$(S,\Delta)$ is then the subgroup of 
$\Perm(S\mid\Delta)\subseteq\Perm(S)$ that preserves $\Delta$. Note 
that $\Perm(S\mid\Delta)$ contains only those maps $f$ preserving $\Delta$ 
whose inverse, $f^{-1}$, also preserve  $\Delta$. Now consider another 
structure, $\Delta'$, and form $\Perm(S\mid\Delta')$. One way in which 
the two structures $\Delta$ and $\Delta'$ may be compared is to 
compare their automorphism groups $\Perm(S\mid\Delta)$ and 
$\Perm(S\mid\Delta')$. Comparing the latter means, in particular, to 
see whether one is contained in the other. Containedness clearly 
defines a partial order relation on the set of subgroups of 
$\Perm(S)$, which we can use to define a partial order on the set of 
structures. One structure, $\Delta$, is said to be strictly 
stronger than (or equally strong as) another structure, $\Delta'$, 
in symbols $\Delta\geq\Delta'$, iff\footnote{Throughout we 
use `iff' as abbreviation for `if and only if'.}  the automorphism 
group of the former is properly contained in (or is equal to) the 
automorphism group of the latter.\footnote{Strictly speaking, 
it would be more appropriate to speak of conjugacy classes of 
subgroups in $\Perm(S)$ here.} In symbols:  
$\Delta\geq\Delta'\Leftrightarrow\Perm(S\mid\Delta)\subseteq\Perm(S\mid\Delta')$.
Note that in this way of speaking a substructure (i.e. one being 
defined by a subset of conditions, relations, objects, etc.) of 
a given structure is said to be weaker than the latter.
This way of thinking of structures in terms of their automorphism
group is adopted from Felix Klein's \emph{Erlanger Programm} 
\cite{Klein:ErlangerProgramm} in which this strategy is used in an 
attempt to classify and compare geometries.  

This general procedure can be applied to Minkowski space, endowed with 
its usual structure (see below). We can than ask whether the automorphism 
group of Minkowski space, which we know is the inhomogeneous Lorentz group
$\ILor$, also called the Poincar\'e group, is already the automorphism 
group of a proper substructure. If this were the case we would say that 
the original structure is redundant. It would then be of interest to try 
and find a minimal set of structures that already imply the Poincar\'e
group. This can be done by trial and error: one starts with 
some more or less obvious substructure, determine its automorphism group, 
and compare it to the Poincar\'e group. Generically it will turn out 
larger, i.e. to properly contain $\ILor$. The obvious questions to ask 
then are:  how much larger? and: what would be a minimal extra 
condition that eliminates the difference? 

\subsection{Definition of Minkowski space and Poincar\'e group}
\label{sec:DefMinkSpacePoinGroup}
These questions have been asked in connection with various 
substructures of Minkowski space, whose definition is as 
follows:
\begin{definition}
\label{def:MinkowskiSpace}
\textbf{Minkowski space} of $n\geq 2$ dimensions, denoted by 
$\Mink{n}$, is a real $n$-dimensional affine space,
whose associated real $n$-dimensional vector space $V$ 
is endowed with a non-degenerate symmetric bilinear form 
$g:V\times V\rightarrow\reals$ of signature $(1,n-1)$
(i.e. there exists a basis $\{e_0,e_1,\cdots,e_{n-1}\}$ of $V$
such that $g(e_a,e_b)=\mathrm{diag}(1,-1,\cdots,-1)$). 
$\Mink{n}$ is also endowed with the standard differentiable
structure of $\reals^{n}$. 
\end{definition}
We refer to Appendix\,\ref{sec:AffineSpaces} for the definition 
of affine spaces. Note also that the last statement concerning 
differentiable structures is put in in view of the strange 
fact that just for the physically most interesting case, $n=4$, 
there exist many inequivalent differentiable structures of 
$\reals^{4}$. Finally we stress that, at this point, we did not 
endow Minkowski space with an orientation or time orientation. 
\begin{definition}
\label{def:PoincareGroup}
The \textbf{Poincar\'e group} in $n\geq 2$ dimensions, 
which is the same as the \textbf{inhomogeneous Lorentz group} 
in $n\geq 2$ dimensions and therefore will be denoted by $\ILor^n$, 
is that subgroup of the general affine group of real $n$-dimensional 
affine space, for which the uniquely associated linear maps 
$f:V\rightarrow V$ are elements of the Lorentz group $\Lor^n$, 
that is, preserve $g$ in the sense that $g\bigl(f(v),f(w)\bigr)=g(v,w)$ 
for all $v,w\in V$.
\end{definition}
See Appendix\,\ref{sec:AffineMaps} for the definition of 
affine maps and the general affine group. Again we stress 
that since we did not endow Minkowski space with any 
orientation, the Poincar\'e group as defined here would not 
respect any such structure. 

As explained in \ref{sec:AffineFrames}, any choice of an 
affine frame  allows us to identify the general affine 
group in $n$ dimensions with the semi-direct product 
$\reals^{n}\rtimes\group{GL}(n)$. That identification 
clearly depends on the choice of the frame. If we 
restrict the bases to those where 
$g(e_a,e_b)=\mathrm{diag}(1,-1,\cdots,-1)$, then  
$\ILor^n$ can be identified with 
$\reals^{n}\rtimes\group{O}(1,n-1)$.     

We can further endow Minkowski space with an \emph{orientation} and, 
independently, a \emph{time orientation}. An orientation of an affine 
space is equivalent to an orientation of its associated vector space 
$V$. A time orientation is also defined trough a time orientation 
of $V$, which is explained below. The subgroup of the Poincar\'e 
group preserving the overall orientation is denoted by $\ILor_+^n$
(proper Poincar\'e group), the one preserving time orientation by 
$\ILor_{\uparrow}^n$ (orthochronous Poincar\'e group), and 
$\ILor_{+\uparrow}^n$ denotes the subgroup preserving both 
(proper orthochronous Poincar\'e group).

Upon the choice of a basis 
we may identify $\ILor_+^n$ with $\reals^{n}\rtimes\group{SO}(1,n-1)$
and $\ILor_{+\uparrow}^n$ with $\reals^{n}\rtimes\group{SO}_0(1,n-1)$,
where $\group{SO}_0(1,n-1)$ is the component of the identity of 
$\group{SO}(1,n-1)$.

Let us add a few more comments about the elementary geometry of 
Minkowski space. We introduce the following notations:
\begin{equation}
\label{eq:def:MinkProd}
v\cdot w\,:=\,g(v,w)\,\qquad\text{and}\qquad
\Vert v\Vert_g:=\sqrt{\vert g(v,v)\vert}\,.
\end{equation} 
We shall also simply write $v^2$ for $v\cdot v$. A vector $v\in V$ 
is called \emph{timelike, lightlike}, or \emph{spacelike} according 
to $v^2$ being $>0$, $=0$, or $<0$ respectively. Non-spacelike vectors 
are also called \emph{causal} and their set, $\mathcal{\bar C}\subset V$, 
is called the \emph{causal-doublecone}. Its interior, $\mathcal{C}$, is 
called the \emph{chronological-doublecone} and its boundary, $\mathcal{L}$,
the \emph{light-doublecone}:    
\begin{subequations}
\label{eq:VecDC}
\begin{alignat}{2}
\label{eq:def:VecCausalDC}
& \mathcal{\bar C}&&:\,=\,\{v\in V\mid v^2\geq 0\}\,,\\
\label{eq:def:VecChronDC}
& \mathcal{C}&&:\,=\,\{v\in V\mid v^2> 0\}\,,\\
\label{eq:def:VecLightDC}
& \mathcal{L}&&:\,=\,\{v\in V\mid v^2=0\}\,.
\end{alignat}
\end{subequations}

A linear subspace $V'\subset V$ is called timelike, lightlike, 
or spacelike according to $g\big\vert_{V'}$ being 
indefinite, negative semi-definite but not negative definite, 
or negative definite respectively. 
Instead of the usual Cauchy-Schwarz-inequality we have 
\begin{subequations}
\label{eq:CauchySchwarz}
\begin{alignat}{3}
\label{eq:CauchySchwarz1}
& v^2w^2&&\,\leq\, (v\cdot w)^2\quad &&\text{for $\Span\{v,w\}$ timelike}\,,\\
\label{eq:CauchySchwarz2}
& v^2w^2&&\,=\,(v\cdot w)^2\quad &&\text{for $\Span\{v,w\}$ lightlike}\,,\\ 
\label{eq:CauchySchwarz3}
& v^2w^2&&\,\geq\, (v\cdot w)^2\quad &&\text{for $\Span\{v,w\}$ spacelike}\,.
\end{alignat}
\end{subequations}

Given a set $W\subset V$ (not necessarily a subspace\footnote{By a 
`subspace' of a vector space we always understand a sub vector-space.}), 
its $g$-orthogonal complement is the subspace
\begin{equation}
\label{eq:OrthCompl}
W^\perp:=\{v\in V\mid v\cdot w=0,\,\forall w\in W\}\,.
\end{equation}
If $v\in V$ is lightlike then $v\in v^\perp$. In fact, $v^\perp$ is 
the unique lightlike hyperplane (cf. Sect.\,\ref{sec:AffineSpaces})
containing $v$. In this case the 
hyperplane $v^\perp$ is called degenerate because the restriction 
of $g$ to $v^\perp$ is degenerate. On the other hand, if $v$ is 
timelike/spacelike $v^\perp$ is spacelike/timelike and 
$v\not\in v^\perp$. Now the hyperplane $v^\perp$ is called 
non-degenerate because the restriction of $g$ to $v^\perp$ is 
non-degenerate.   

Given any subset $W\subset V$, we can attach it to a point $p$ 
in $\mathbb{M}^n$:
\begin{equation}
\label{eq:AffineSubset}
W_p:=p+W:=\{p+w\mid w\in W\}\,.
\end{equation}
In particular, the causal-, chronological-, and light-doublecones at 
$p\in\mathbb{M}^n$ are given by: 
\begin{subequations}
\label{eq:AffDC}
\begin{alignat}{2}
\label{eq:def:AffCausalDC}
& \mathcal{\bar C}_p&&:\,=\,p+\mathcal{\bar C}\,,\\
\label{eq:def:AffPointChronDC}
& \mathcal{C}_p&&:\,=\,p+\mathcal{C}\,,\\
\label{eq:def:AffLightDC}
& \mathcal{L}_p&&:\,=\,p+\mathcal{L}\,.
\end{alignat}
\end{subequations}
If $W$ is a subspace of $V$ then $W_p$ is an affine subspace of 
$\mathbb{M}^n$ over $W$. If $W$ is time-, light-, or spacelike 
then $W_p$ is also called time-, light-, or spacelike. 
Of particular interest are the hyperplanes $v_p^\perp$
which are timelike, lightlike, or spacelike according to $v$ 
being spacelike, lightlike, or timelike respectively. 

Two points $p,q\in\mathbb{M}^n$ are said to be timelike-, lightlike-, 
or spacelike separated if the line joining them (equivalently: the 
vector $p-q$) is timelike, lightlike, or spacelike respectively. 
Non-spacelike separated points are also called causally separated and 
the line though them is called a causal line.  

It is easy to show that the relation $v\sim w\Leftrightarrow v\cdot w>0$ 
defines an equivalence relation (cf. Sect.\,~\ref{sec:GroupActions}) on 
the set of timelike vectors. (Only transitivity is non-trivial, i.e. if 
$u\cdot v>0$ and $v\cdot w>0$ then $u\cdot w>0$. To show this, decompose 
$u$ and $w$ into their components parallel and perpendicular to $v$.) 
Each of the two equivalence classes is a \emph{cone} in $V$, that is, 
a subset closed under addition and multiplication with positive numbers. 
Vectors in the same class are said to have the same time orientation. 
In the same fashion, the relation $v\sim w\Leftrightarrow v\cdot w\geq 0$ 
defines an equivalence relation on the set of causal vectors, with both 
equivalence classes being again cones. The existence of these 
equivalence relations is expressed by saying that $\mathbb{M}^n$ is 
\emph{time orientable}. Picking one of the two possible time 
orientations is then equivalent to specifying a single timelike 
reference vector, $v_*$, whose equivalence class of directions may 
be called the \emph{future}. This being done we can speak of the future 
(or forward, indicated by a superscript $+$) and past (or backward, 
indicated by a superscript $-$) cones:  
\begin{subequations}
\label{eq:VecC}
\begin{alignat}{2}
\label{eq:def:VecCausalC}
& \mathcal{\bar C}^\pm &&:\,=\,
  \{v\in\mathcal{\bar C}\mid v\cdot v_*\gtrless 0\}\,,\\
\label{eq:def:VecChronC}
& \mathcal{C}^\pm &&:\,=\,
  \{v\in\mathcal{C}\mid v\cdot v_*\gtrless 0\}\,,\\
\label{eq:def:VecLightC}
& \mathcal{L}^\pm &&:\,=\,
  \{v\in \mathcal{L}\mid v\cdot v_*\gtrless 0\}\,.
\end{alignat}
\end{subequations}
Note that $\mathcal{\bar C}^\pm=\mathcal{C}^\pm\cup\mathcal{L}^\pm$
and $\mathcal{C}^\pm\cap\mathcal{L}^\pm=\emptyset$. Usually 
$\mathcal{L}^+$ is called the future and $\mathcal{L}^-$ the past 
lightcone. Mathematically speaking this is an abuse of language 
since, in contrast to $\mathcal{\bar C}^\pm$ and $\mathcal{C}^\pm$, 
they are not cones: They are each invariant (as sets) under 
multiplication with positive real numbers, but adding to vectors in 
$\mathcal{L}^\pm$ will result in a vector in $\mathcal{C}^\pm$ unless 
the vectors were parallel.

As before, these cones can be attached to the points in $\mathbb{M}^n$.
We write in a straightforward manner: 
\begin{subequations}
\label{eq:AffC}
\begin{alignat}{2}
\label{eq:def:AffCausalC}
& \mathcal{\bar C}_p^\pm &&:\,=\,p+\mathcal{\bar C}^\pm\,,\\
\label{eq:def:AffChronC}
& \mathcal{C}_p^\pm &&:\,=\,p+\mathcal{C}^\pm\,,\\
\label{eq:def:AffLightC}
& \mathcal{L}_p^\pm &&:\,=\,p+\mathcal{L}^\pm\,.
\end{alignat}
\end{subequations}

The Cauchy-Schwarz inequalities (\ref{eq:CauchySchwarz}) result 
in various generalised triangle-inequalities. Clearly, for spacelike 
vectors, one just has the ordinary triangle inequality. But for causal or 
timelike vectors one has to distinguish the cases according to the relative 
time orientations. For example, for timelike vectors of equal time 
orientation, one obtains the reversed triangle inequality:
\begin{equation}
\label{sec:InvTriangleIneq}
\Vert v+w\Vert_g\geq\Vert v\Vert_g +\Vert w\Vert_g\,,
\end{equation}
with equality iff $v$ and $w$ are parallel. It expresses the geometry 
behind the `twin paradox'. 

Sometimes a Minkowski `distance function' 
$d:\Mink{n}\times\Mink{n}\rightarrow\reals$ is introduced through 
\begin{equation}
\label{eq:Def:MinkowskiDistance}
d(p,q):=\Vert p-q\Vert_g\,. 
\end{equation}
Clearly this is not a distance function in the ordinary sense, 
since it is neither true that $d(p,q)=0\Leftrightarrow p=q$ nor 
that $d(p,w)+d(w,q)\geq d(p,q)$ for all $p,q,w$. 

\subsection{From metric to affine structures}
\label{sec:FromMetricToAffine}
In this section we consider general isometries of Minkowski 
space. By this we mean general bijections 
$F:\Mink{n}\rightarrow\Mink{n}$  (no requirement 
like continuity or even linearity is made) which preserve 
the Minkowski distance (\ref{eq:Def:MinkowskiDistance}) 
as well as the time or spacelike character; hence  
\begin{equation}
\label{eq:Def:MinkowskiIsometries}
\bigl(F(p)-F(q)\bigr)^2=(p-q)^2\qquad
\text{for all}\quad p,q\in\Mink{n}\,.
\end{equation}
Poincar\'e transformations form a special class of such 
isometries, namely those which are affine. Are there non-affine 
isometries? One might expect a whole Pandora's box full of 
wild (discontinuous) ones. But, fortunately, they do not exist:
Any map $f:V\rightarrow V$ satisfying $(f(v))^2=v^2$
for all $v$ must be linear. As a warm up, we show  
\begin{theorem}
\label{thm:NoWildIsometries}
Let $f:V\rightarrow V$ be a surjection (no further conditions) 
so that $f(v)\cdot f(w)=v\cdot w$ for all $v,w\in V$, then $f$ 
is linear. 
\end{theorem}
\begin{proof}
Consider $I:=\bigl(af(u)+bf(v)-f(au+bv)\bigr)\cdot w$. 
Surjectivity allows to write $w=f(z)$, so that 
$I=a\,u\cdot z+b\,v\cdot z-(au+bv)\cdot z$, which vanishes 
for all $z\in V$. Hence $I=0$ for all $w\in V$, which by non-degeneracy 
of $g$ implies the linearity of $f$.
\end{proof}
This shows in particular that any bijection 
$F:\Mink{n}\rightarrow\Mink{n}$ of Minkowski space whose associated 
map $f:V\rightarrow V$, defined by $f(v):=F(o+v)-F(o)$ for some chosen 
basepoint $o$, preserves the Minkowski metric must be a Poincar\'e 
transformation. As already indicated, this result can be considerably 
strengthened. But before going into this, we mention a special and 
important class of linear isometries of $(V,g)$, namely reflections 
at non-degenerate hyperplanes. The reflection at $v^\perp$ is defined 
by 
\begin{equation}
\label{eq:Reflection1}
\rho_v(x):=x-2\,v\ \frac{x\cdot v}{v^2}\,.
\end{equation}
Their significance is due to the following 
\begin{theorem}[Cartan, Dieudonn\'e]
Let the dimension of $V$ be $n$. Any isometry of 
$(V,g)$ is the composition of at most $n$ reflections. 
\end{theorem}
\begin{proof}
Comprehensive proofs may be found in \cite{Jacobson:BasicAlgebraI}
or \cite{Berger:Geometry2}. The easier proof for at most $2n-1$ 
reflections is as follows: Let $\phi$ be a linear isometry and 
$v\in V$ so that $v^2\ne 0$ (which certainly exists). Let 
$w=\phi(v)$, then $(v+w)^2+(v-w)^2=4v^2\ne 0$ so that $w+v$ and $w-v$ 
cannot simultaneously have zero squares. So let $(v\mp w)^2\ne 0$
(understood as alternatives), then $\rho_{v\mp w}(v)=\pm w$ and 
$\rho_{v\mp w}(w)=\pm v$. Hence $v$ is eigenvector with eigenvalue 
$1$ of the linear isometry given by   
\begin{equation}
\label{eq:proof:CartanDieu1}  
\phi'=
\begin{cases}
\rho_{v-w}\circ \phi &\text{if}\ (v-w)^2\ne 0\,,\\
\rho_v\circ\rho_{v+w}\circ \phi &\text{if}\ (v-w)^2= 0\,.
\end{cases}
\end{equation}
Consider now the linear isometry $\phi'\big\vert_{v^\perp}$
on $v^\perp$ with induced bilinear form $g\big\vert_{v^\perp}$,
which is non-degenerated due to $v^2\ne 0$. We conclude by 
induction: At each dimension we need at most two reflections to 
reduce the problem by one dimension. After $n-1$ steps we have 
reduced the problem to one dimension, where we need at most 
one more reflection. Hence we need at most $2(n-1)+1=2n-1$ 
reflections which upon composition with $\phi$ produce the identity. 
Here we use that any linear isometry in $v^\perp$ can be canonically 
extended to $\Span\{v\}\oplus v^\perp$ by just letting it act 
trivially on $\Span\{v\}$.
\end{proof}
Note that this proof does not make use of the signature of $g$. In fact,
the theorem is true for any signatures; it only depends on $g$ being
symmetric and non degenerate.

\subsection{From causal to affine structures}
\label{sec:FromCausalToAffine}
As already mentioned, Theorem\,\ref{thm:NoWildIsometries} can be 
improved upon, in the sense that the hypothesis for the map being 
an isometry is replaced by the hypothesis that it merely preserve 
some relation that derives form the metric structure, but is not 
equivalent to it. In fact, there are various such relations which 
we first have to introduce. 

The family of cones $\{\mathcal{\bar C}_q^+\mid q\in\mathbb{M}^n\}$ 
defines a partial-order relation (cf.~Sect.\,\ref{sec:GroupActions}), 
denoted by $ \geq$, on spacetime as follows:
$p\geq q$ iff $p\in\mathcal{\bar C}^+_q$, i.e. iff $p-q$ is causal 
and future pointing. Similarly, the family 
$\{\mathcal{C}^+_q\mid q\in\mathbb{M}^n\}$ defines a strict partial 
order, denoted by $>$, as follows: $p>q$ iff 
$p\in\mathcal{C}^+_q$, i.e. if $p-q$ is timelike and future pointing. 
There is a third relation, called $\gtrdot$, defined as follows: 
$p\gtrdot q$ iff $p\in\mathcal{L}_q^+$, i.e. $p$ is on the future 
lightcone at $q$. It is not a partial order due to the lack 
of transitivity, which, in turn, is due to the lack of the lightcone 
being a cone (in the proper mathematical sense explained above). 
Replacing the future ($+$) with the past ($-$) cones gives the 
relations $\leq$, $<$, and $\lessdot$.

It is obvious that the action of $\ILor^\uparrow$ (spatial reflections 
are permitted) on $\mathbb{M}^n$ maps each of the six families of 
cones (\ref{eq:AffC}) into itself and therefore leave each of the 
six relations invariant. For example: Let $p>q$ and $F\in\ILor^\uparrow$, 
then $(p-q)^2>0$ and $p-q$ future pointing, but also $(F(p)-F(q))^2>0$
and $F(p)-F(q)$ future pointing, hence $F(p)>F(q)$. Another set of `obvious' 
transformations of $\mathbb{M}^n$ leaving these relations invariant is 
given by all dilations: 
\begin{equation}
\label{eq:def:Dilations}
d_{(\lambda,m)}:\mathbb{M}^n\rightarrow\mathbb{M}^n\,,\quad
p\mapsto d_{(\lambda,m)}(p):=\lambda(p-m)+m\,,
\end{equation}
where $\lambda\in\mathbb{R }_+$ is the constant dilation-factor and 
$m\in\mathbb{M}^n$ the centre. This follows from  
$\bigl(d_{\lambda,m}(p)-d_{\lambda,m}(q)\bigr)^2=\lambda^2(p-q)^2$,
$\bigl(d_{\lambda,m}(p)-d_{\lambda,m}(q)\bigr)\cdot v_*=\lambda(p-q)\cdot v_*$,
and the positivity of $\lambda$. Since translations are already 
contained in $\ILor^\uparrow$, the group generated by $\ILor^\uparrow$
and all $d_{\lambda,m}$ is the same as the group generated by 
$\ILor^\uparrow$ and all $d_{\lambda,m}$ for fixed $m$. 

A seemingly difficult question is this: What are the most general 
transformations of $\mathbb{M}^n$ that preserve those relations? 
Here we understand `transformation' synonymously with `bijective map', 
so that each transformation $f$ has in inverse $f^{-1}$. `Preserving 
the relation' is taken to mean that $f$ \emph{and} $f^{-1}$ 
preserve the relation. Then the somewhat surprising answer to the 
question just posed is that, in three or more spacetime dimensions, 
there are no other such transformations besides those already listed: 

\begin{theorem}
Let $\succ$ stand for any of the relations $\geq,>,\gtrdot$ and 
let $F$ be a bijection of $\mathbb{M}^n$ with $n\geq 3$, such that 
$p\succ q$  implies  $F(p)\succ F(q)$ and 
$F^{-1}(p)\succ F^{-1}(q)$. Then $F$ is the composition of an 
Lorentz transformation in $\ILor^\uparrow$ with a dilation.
\end{theorem}
\begin{proof}
These results were proven by A.D.\,Alexandrov and independently 
by E.C. Zeeman. A good review of Alexandrov's results is 
\cite{Alexandrov:1975}; Zeeman's paper is \cite{Zeeman:1964}.
The restriction to $n\geq 3$ is indeed necessary, as for $n=2$ 
the following possibility exists: Identify $\mathbb{M}^2$
with $\reals^2$ and the bilinear form $g(z,z)=x^2-y^2$,
where $z=(x,y)$. Set $u:=x-y$ and $v:=x+y$ and define 
$f:\reals^2\rightarrow \reals^2$ by $f(u,v):=(h(u),h(v))$,
where $h:\reals\rightarrow\reals$ is any smooth function 
with $h'>0$. This defines an orientation preserving diffeomorphism
of $\reals^2$ which transforms the set of lines $u=$ const. and 
$v=$ const. respectively into each other. Hence it preserves the 
families of cones (\ref{eq:def:AffCausalC}). Since these 
transformations need not be affine linear they are not generated by 
dilations and Lorentz transformations.
\end{proof}

These results may appear surprising since without a continuity 
requirement one might expect all sorts of wild behaviour to allow 
for more possibilities. However, a little closer inspection reveals 
a fairly obvious reason for why continuity is implied here.  
Consider the case in which a transformation $F$ preserves the 
families $\{\mathcal{C}^+_q\mid q\in\mathbb{M}^n\}$ and 
$\{\mathcal{C}^-_q\mid q\in\mathbb{M}^n\}$. The open 
diamond-shaped sets (usually just called `open diamonds'),
\begin{equation}
\label{eq:OpenDiamonds}
U(p,q):=(\mathcal{C}^+_p\cap\mathcal{C}^-_q)\cup
        (\mathcal{C}^+_q\cap\mathcal{C}^-_p)\,,
\end{equation}
are obviously open in the standard topology of $\mathbb{M}^n$ 
(which is that of $\reals^n$). Note that at least one of the 
intersections in (\ref{eq:OpenDiamonds}) is always empty. 
Conversely, is is also easy to see that each open set of 
$\mathbb{M}^n$ contains an open diamond. Hence the topology 
that is defined by taking the $U(p,q)$ as sub-base (the basis 
being given by their finite intersections) is equivalent to the 
standard topology of $\mathbb{M}^n$. But, by hypothesis, $F$ and 
$F^{-1}$ preserves the cones $\mathcal{C}^\pm_q$ and therefore 
open sets, so that $F$ must, in fact, be a homeomorphism. 

There is no such \emph{obvious} continuity input if one 
makes the strictly weaker requirement that instead of the cones 
(\ref{eq:AffC}) one only preserves the doublecones 
(\ref{eq:AffDC}). Does that allow for more transformations, except 
for the obvious time reflection? The answer is again in the negative.
The following result was shown by  Alexandrov (see his review 
\cite{Alexandrov:1975}) and later, in a different fashion, by 
Borchers and Hegerfeld~\cite{Borchers.Hegerfeld:1972}: 
\begin{theorem}
\label{thm:BorchersHegerfeld}
Let $\sim$ denote any of the relations: $p\sim q$ iff $(p-q)^2\geq 0$,
 $p\sim q$ iff $(p-q)^2> 0$, or $p\sim q$ iff $(p-q)^2=0$. Let $F$ be a 
bijection of $\mathbb{M}^n$ with $n\geq 3$, such that $p\sim q$  
implies  $F(p)\sim F(q)$ and $F^{-1}(p)\sim F^{-1}(q)$. Then $F$ 
is the composition of an Lorentz transformation in $\ILor$ with a 
dilation.
\end{theorem}

All this shows that, up to dilations, Lorentz transformations can 
be characterised by the causal structure of Minkowski space. 
Let us focus on a particular sub-case of 
Theorem\,\ref{thm:BorchersHegerfeld}, which says that any bijection 
$F$ of $\mathbb{M}^n$ with $n\geq 3$, which satisfies 
$\Vert p-q\Vert_g=0\Leftrightarrow\Vert F(p)-F(q)\Vert_g=0$ must be 
the composition of a dilation and a transformation in $\ILor$. 
This is sometimes referred to as \emph{Alexandrov's theorem}. It gives a 
precise answer to the following physical question: To what extent does 
the principle of the constancy of a finite speed of light \emph{alone} 
determine the relativity group? The answer is, that it determines 
it to be a subgroup of the 11-parameter group of Poincar\'e 
transformations and constant rescalings, which is as close to 
the Poincar\'e group as possibly imaginable. 

Alexandrov's Theorem is, to my knowledge, the closest analog in 
Minkowskian geometry to the famous theorem of Beckman and Quarles 
\cite{Beckman.Quarles:1953}, which refers to Euclidean geometry 
and reads as follows\footnote{In fact, Beckman and Quarles proved 
the conclusion of Theorem\,\ref{thm:BeckmanQuarles} under slightly 
weaker hypotheses: They allowed the map $f$ to be `many-valued', 
that is, to be a map $f:\reals^n\rightarrow\mathcal{S}^n$, 
where $\mathcal{S}^n$ is the set of non-empty subsets of $\reals^n$,
such that $\Vert x-y\Vert=\delta\Rightarrow \Vert x'-y'\Vert=\delta$
for any $x'\in f(x)$ and any $y'\in f(y)$. However, given the 
statement of Theorem\,\ref{thm:BeckmanQuarles}, it is immediate 
that such `many-valued maps' must necessarily be single-valued. To see 
this, assume that $x_*\in\reals^n$ has the two image points 
$y_1,y_2$ and define $h_i:\reals^n\rightarrow\reals^n$ for 
$i=1,2$ such that $h_1(x)=h_2(x)\in f(x)$ for all $x\ne x_*$ and 
$h_i(x_*)=y_i$. Then, according to Theorem\,\ref{thm:BeckmanQuarles},
$h_i$ must both be Euclidean motions. Since they are continuous and 
coincide for all $x\ne x_*$, they must also coincide at $x_*$.}:

\begin{theorem}[Beckman and Quarles 1953]
\label{thm:BeckmanQuarles}
Let $\reals^n$ for $n\geq 2$ be endowed with the standard Euclidean 
inner product $\langle\cdot\mid\cdot\rangle$. The associated norm is 
given by $\Vert x\Vert:=\sqrt{\langle x\mid x\rangle}$. 
Let $\delta$ be any fixed positive real number and 
$f:\reals^n\rightarrow\reals^n$ any map such that 
$\Vert x-y\Vert=\delta\Rightarrow \Vert f(x)-f(y)\Vert=\delta$; then 
$f$ is a Euclidean motion, i.e. $f\in\reals^n\rtimes\group{O}(n)$.
\end{theorem}
Note that there are three obvious points which let the result of 
Beckman and Quarles in Euclidean space appear somewhat stronger than 
the theorem of Alexandrov in Minkowski space: 
\begin{itemize}
\item[1.]
The conclusion of Theorem\,\ref{thm:BeckmanQuarles} holds for 
\emph{any} $\delta\in\reals_+$, whereas Alexandrov's theorem 
singles out lightlike distances.
\item[2.]
In Theorem\,\ref{thm:BeckmanQuarles}, $n=2$ is not excluded. 
\item[3.]
In Theorem\,\ref{thm:BeckmanQuarles}, $f$ is not required to be a 
bijection, so that we did not assume the existence of an inverse map 
$f^{-1}$. Correspondingly, there is no assumption 
that $f^{-1}$ also preserves the distance $\delta$. 
\end{itemize}

\subsection{The impact of the law of inertia}
\label{sec:InertiaAndAffineStructure}
In this subsection we wish to discuss the extent to which 
the law of inertia already determines the automorphism group 
of spacetime. 

The law of inertia privileges a subset of paths in spacetime 
form among all paths; it defines a so-called 
\emph{path structure}~\cite{Ehlers.Koehler:1977}\cite{Coleman.Korte:1980}. 
These privileged paths correspond to the motions of privileged 
objects called \emph{free particles}. The existence of such 
privileged objects is by no means obvious and must be taken as 
a contingent and particularly kind property of nature. It has 
been known for long \cite{Lange:1885}\cite{Thomson:1884}\cite{Tait:1884} 
how to operationally construct timescales and spatial reference 
frames relative to which free particles will move uniformly and on 
straight lines respectively---all of them! (A summary of these papers 
is given in \cite{Giulini:2002b}.) 
These special timescales and spatial reference frames were termed 
\emph{inertial} by Ludwig Lange~\cite{Lange:1885}. Their existence 
must again be taken as a very particular and very kind feature of 
Nature. Note that `uniform in time' and `spatially straight' 
together translate to `straight in spacetime'. We also emphasise 
that `straightness' of ensembles of paths can be characterised 
intrinsically, e.g., by the Desargues property~\cite{Pfister:2004}.
All this is true if free particles are given. We do not discuss 
at this point whether and how one should characterise them 
independently (cf.~\cite{Frege:1891}). 

The spacetime structure so defined is usually referred to as 
\emph{projective}. It it not quite that of an affine space, since the 
latter provides in addition each straight line with a distinguished 
two-parameter family of parametrisations, corresponding to a notion 
of \emph{uniformity} with which the line is traced through. Such a 
privileged parametrisation of spacetime paths is not provided by 
the law of inertia, which only provides privileged parametrisations 
of spatial paths, which we already took into account in the 
projective structure of spacetime. Instead, an affine structure of 
spacetime may once more be motivated by another contingent property 
of Nature, shown by the existence of elementary clocks (atomic 
frequencies) which do define the same uniformity structure on inertial 
world lines---all of them! Once more this is a highly non-trivial and 
very kind feature of Nature. In this way we would indeed arrive at the 
statement that spacetime is an affine space. However, as we shall 
discuss in this subsection, the affine group already emerges as 
automorphism group of inertial structures without the introduction 
of elementary clocks. 

First we recall the main theorem  of affine geometry. For that we 
make the following 
\begin{definition}
\label{def:Collineation}
Three points in an affine space are called \textbf{collinear}
iff they are contained in a single line. A map between 
affine spaces is called a \textbf{collineation} iff it maps 
each triple of collinear points to collinear points.
\end{definition}
Note that in this definition no other condition is 
required of the map, like, e.g., injectivity. The main 
theorem now reads as follows:
\begin{theorem}
\label{thm:CollAreAffine}
A bijective collineation of a real affine space of 
dimension $n\geq 2$ is necessarily an affine map. 
\end{theorem}
A proof may be found in \cite{Berger:Geometry1}. That the 
theorem is non-trivial can, e.g., be seen from the fact 
that it is not true for complex affine spaces. The crucial 
property of the real number field is that it does 
not allow for a non-trivial automorphisms (as field).   

A particular consequence of Theorem\,\ref{thm:CollAreAffine}
is that bijective collineation are necessarily continuous 
(in the natural topology of affine space). This is of 
interest for the applications we have in mind for the 
following reason: Consider the set $P$ of all lines in 
some affine space $S$. $P$ has a natural topology induced 
from $S$. Theorem\,\ref{thm:CollAreAffine} now implies that 
bijective collineations of $S$ act as homeomorphism of $P$. 
Consider an open subset $\Omega\subset P$ and the subset of all 
collineations that fix $\Omega$ (as set, not necessarily its 
points). Then these collineations also fix the boundary 
$\partial\Omega$ of $\Omega$ in $P$. For example, if 
$\Omega$ is the set of all timelike lines in Minkowski 
space, i.e., with a slope less than some chosen value 
relative to some fixed direction, then it follows that 
the bijective collineations which together with their 
inverse map timelike lines to timelike lines also maps the 
lightcone to the lightcone. It immediately follows that 
it must be the composition of a Poincar\'e transformation
as a constant dilation. Note that this argument also works 
in two spacetime dimensions, where the Alexandrov-Zeeman 
result does not hold. 

The application we have in mind is to inertial motions, 
which are given by lines in affine space. In that respect 
Theorem\,\ref{thm:CollAreAffine} is not quite appropriate. 
Its hypotheses are weaker than needed, insofar as it would 
suffice to require straight lines to be mapped to straight 
lines. But, more importantly, the hypotheses are also stronger 
than what seems physically justifiable, insofar as not every line 
is realisable by an inertial motion.  In particular, one would 
like to know whether Theorem\,\ref{thm:CollAreAffine} can still be 
derived by restricting to \emph{slow collineations}, which one 
may define by the property that the corresponding lines should 
have a slope less than some non-zero angle (in whatever measure,
as long as the set of slow lines is open in the set of all lines) 
from a given (time-)direction. This is indeed the case, as one 
may show from going through the proof of 
Theorem\,\ref{thm:CollAreAffine}. Slightly easier to prove is the 
following: 
\begin{theorem}
\label{thm:TimelikeCollAreAffine}
Let $F$ be a bijection of real $n$-dimensional affine space 
that maps slow lines to slow lines, then $F$ is an affine map. 
\end{theorem}
A proof may be found in \cite{Goldstein:2007}.
If `slowness' is defined via the lightcone of a Minkowski 
metric $g$, one immediately obtains the result that the 
affine maps must be composed from Poincar\'e transformations 
and dilations. The reason is 
\begin{lemma}
\label{thm:SameKernelPropto}
Let $V$ be a finite dimensional real vector space of dimension 
$n\geq 2$ and $g$ be a non-degenerate symmetric bilinear form 
on $V$ of signature $(1,n-1)$. Let $h$ be any other 
symmetric bilinear form on $V$. The `light cones' for both
forms are defined by $\mathcal{L}_g:=\{v\in V\mid g(v,v)=0\}$
and $\mathcal{L}_h:=\{v\in V\mid h(v,v)=0\}$. Suppose  
$\mathcal{L}_g\subseteq\mathcal{L}_h$, then $h=\alpha\,g$ for 
some $\alpha\in\reals$.
\end{lemma}
\begin{proof}
Let $\{e_0,e_1,\cdots ,e_{n-1}\}$ be a basis of $V$ such that 
$g_{ab}:=g(e_a,e_b)=\mathrm{diag}(1,-1,\cdots,-1)$. Then 
$(e_0\pm e_a)\in\mathcal{L}_g$ for $1\leq a\leq n-1$ implies
(we write $h_{ab}:=h(e_a,e_b)$): 
$h_{0a}=0$ and $h_{00}+h_{aa}=0$. Further, 
$(\sqrt{2}e_0+e_a+e_b)\in\mathcal{L}_g$ for $1\leq a<b\leq n-1$
then implies $h_{ab}=0$ for $a\ne b$. Hence $h=\alpha\,g$ with 
$\alpha=h_{00}$.
\end{proof}
This can be applied as follows: If  $F:S\rightarrow S$ is affine 
and maps lightlike lines to lightlike lines, then the associated 
linear map $f:V\rightarrow V$ maps lightlike vectors to lightlike 
vectors. Hence $h(v,v):=g(f(v),f(v))$ vanishes if $g(v,v)$ vanishes 
and therefore $h=\alpha g$ by Lemma\,\ref{thm:SameKernelPropto}. 
Since $f(v)$ is timelike if $v$ is timelike, $\alpha$ is positive. 
Hence we may define $f':=f/\sqrt{\alpha}$ and have 
$g(f'(v),f'(v))=g(v,v)$ for all $v\in v$, saying that $f'$ is a 
Lorentz transformation. $f$ is the composition of a Lorentz 
transformation and a dilation by $\sqrt{\alpha}$.

\subsection{The impact of relativity}
\label{sec:ImpactOfRelativity}
As is well known, the two main ingredients in Special Relativity 
are the Principle of Relativity (henceforth abbreviated by PR)
and the principle of the constancy of light. We have seen above 
that, due to Alexandrov's Theorem, the latter almost suffices to 
arrive at the Poincar\'e group. In this section we wish to address
the complementary question: Under what conditions and to what extent 
can the RP \emph{alone} justify the Poincar\'e group? 

This question was first addressed by Ignatowsky~\cite{Ignatowsky:1910}, 
who showed that under a certain set of technical assumptions (not 
consistently spelled out by him) the RP alone suffices to arrive at 
a spacetime symmetry group which is either the inhomogeneous Galilei 
or the inhomogeneous Lorentz group, the latter for some yet undetermined 
limiting velocity $c$. 

More precisely, what is actually shown in this fashion is, as we will 
see, that the relativity group must contain either the proper 
orthochronous Galilei or Lorentz group, if the group is required to 
comprise at least spacetime translations, spatial rotations, and 
boosts (velocity transformations). What we hence gain is the 
group-theoretic insight of how these transformations must combine 
into a common group, given that they form a group at all. We do 
not learn anything about other transformations, like spacetime 
reflections or dilations, whose existence we neither required nor 
ruled out at this level. 

The work of Ignatowsky was put into a logically more coherent form
by Franck \& Rothe~\cite{Frank.Rothe:1911}\cite{Frank.Rothe:1912}, 
who showed that some of the technical assumptions could be dropped. 
Further formal simplifications were achieved by Berzi \& Gorini 
\cite{Berzi.Gorini:1969}. Below we shall basically follow their line 
of reasoning, except that we do not impose the continuity of the 
transformations as a requirement, but conclude it from their 
preservation of the inertial structure plus bijectivity. See also 
\cite{Bacry.Levy-Leblond:1968} for an alternative discussion on 
the level of Lie algebras.

For further determination of the automorphism group of spacetime 
we invoke the following principles:
\begin{itemize}
\item[ST1:]\hspace{0.6cm}
Homogeneity of spacetime.
\item[ST2:]\hspace{0.6cm}
Isotropy of space.
\item[ST3:]\hspace{0.6cm}
Galilean principle of relativity.
\end{itemize}
We take ST1 to mean that the sought-for group should include all 
translations and hence be a subgroup of the general affine group. 
With respect to some chosen basis, it must be of the form 
$\reals^4\rtimes\group{G}$, where $\group{G}$ is a subgroup of 
$\group{GL}(4,\reals)$. ST2 is interpreted as saying that $G$ should 
include the set of all spatial rotations. If, with respect to some 
frame, we write the general element $A\in\group{GL}(4,\reals)$ in 
a $1+3$ split form (thinking of the first coordinate as time, the 
other three as space), we want $\group{G}$ to include all       
\begin{equation}
\label{eq:SO3Imbedding}
R(\mat{D})=
\begin{pmatrix}
1&\vec 0^\top\\
\vec 0&\mat{D}
\end{pmatrix}\,,\qquad
\text{where}\quad
\mat{D}\in\group{SO}(3)\,.
\end{equation}
Finally, ST3 says that velocity transformations, henceforth 
called `boosts', are also contained in $\group{G}$. However, at this 
stage we do not know how boosts are to be represented mathematically. 
Let us make the following assumptions: 
\begin{itemize}
\item[B1:]\hspace{0.2cm}
Boosts $B(\vec v)$ are labelled by a vector 
$\vec v\in B_c(\reals^3)$, where $B_c(\reals^3)$ is the 
open ball in $\reals^3$ of radius $c$. The physical interpretation 
of $\vec v$ shall be that of the boost velocity, as measured in the system 
from which the transformation is carried out. We allow $c$ to be 
finite or infinite ($B_\infty(\reals^3)=\reals^3$). 
$\vec v=\vec 0$ corresponds to the identity transformation, 
i.e. $B(\vec 0)=\text{id}_{\reals^4}$. We also assume that 
$\vec v$, considered as coordinate function on the group, 
is continuous. 
\item[B2:]\hspace{0.2cm}
As part of ST2 we require equivariance of boosts under rotations:
\begin{equation}
\label{eq:EquivBoostsRot}
R(\mat{D})\cdot B(\vec v)\cdot R(\mat{D}^{-1})=
B(\mat{D}\cdot\vec v)\,.
\end{equation}
\end{itemize}
The latter assumption allows us to restrict attention to boost 
in a fixed direction, say  that of the positive $x$-axis. 
Once their analytical form is determined
as function of $v$, where $\vec v=v\vec e_x$, we deduce the 
general expression for boosts using (\ref{eq:EquivBoostsRot}) and 
(\ref{eq:SO3Imbedding}). We make no assumptions involving space 
reflections.\footnote{Some derivations in the literature of the 
Lorentz group do not state the equivariance property 
(\ref{eq:EquivBoostsRot}) explicitly, though they all use it
(implicitly), usually in statements to the effect that it is 
sufficient to consider boosts in one fixed direction. Once this 
restriction is effected, a one-dimensional spatial reflection 
transformation is considered to relate a boost transformation to 
that with opposite velocity. This then gives the impression that 
reflection equivariance is also invoked, though this is not 
necessary in spacetime dimensions greater than two, for 
(\ref{eq:EquivBoostsRot}) allows to invert one axis through a 
180-degree rotation about a perpendicular one.} 
We now restrict attention to $\vec v=v\vec e_x$. We wish to determine
the most general form of $B(\vec v)$ compatible with all requirements 
put so far. We proceed in several steps: 
\begin{enumerate}
\item
Using an arbitrary rotation $\mat{D}$ around the $x$-axis, 
so that $\mat{D}\cdot\vec v=\vec v$, equation (\ref{eq:EquivBoostsRot}) 
allows to prove that
\begin{equation}
\label{eq:FormBoost1}
B(v\vec e_x)=
\begin{pmatrix}
\mat{A}(v)&0\\
0&\alpha(v)\mat{1}_2
\end{pmatrix}\,,
\end{equation}
where here we wrote the $4\times 4$ matrix in a $2+2$ decomposed form. 
(i.e. $\mat{A}(v)$ is a $2\times 2$ matrix and $\mat{1}_2$ is the 
 $2\times 2$ unit-matrix). Applying (\ref{eq:EquivBoostsRot}) once 
more, this time using a $\pi$-rotation about the $y$-axis, we learn 
that $\alpha$ is an even function, i.e. 
\begin{equation}
\label{eq:FormBoosts2}
\alpha(v)=\alpha(-v)\,.
\end{equation}
Below we will see that $\alpha(v)\equiv 1$. 

\item
Let us now focus on $\mat{A }(v)$, which defines the action of 
the boost in the $t-x$ plane. We write 
\begin{equation}
\label{eq:FormBoost3}
\begin{pmatrix}t\\x\end{pmatrix}
\mapsto
\begin{pmatrix}t'\\x'\end{pmatrix}
=\mat{A}(v)\cdot\begin{pmatrix}t\\x\end{pmatrix}=
\begin{pmatrix}a(v)&b(v)\\c(v)&d(v)\end{pmatrix}
\cdot\begin{pmatrix}t\\x\end{pmatrix}\,.
\end{equation}
We refer to the system with coordinates $(t,x)$ as $K$ and that 
with coordinates $(t',x')$ as $K'$. From (\ref{eq:FormBoost3})
and the inverse (which is elementary to compute) one infers that 
the velocity $v$ of $K'$ with respect to $K$ and the velocity $v'$ 
of $K$ with respect to $K'$ are given by 
\begin{subequations}
\label{eq:FormBoost4}
\begin{alignat}{3}
\label{eq:FormBoost4a}     
& v\,&&=\,&&-\,c(v)/d(v)\,,\\
\label{eq:FormBoost4b} 
& v'\,&&=\,&&-\,v\,d(v)/a(v)\,=:\,\varphi(v)\,.
\end{alignat}
\end{subequations}
Since the transformation $K'\rightarrow K$ is the inverse of 
$K\rightarrow K'$, the function $\varphi:(-c,c)\rightarrow (-c,c)$ 
obeys 
\begin{equation}
\label{eq:FormBoost5}
\mat{A}(\varphi(v))=(\mat{A}(v))^{-1}\,.
\end{equation}
Hence $\varphi$ is a bijection of the open interval $(-c,c)$ 
onto itself and obeys 
\begin{equation}
\label{eq:FormBoost6}
\varphi\circ\varphi=\text{id}_{(-c,c)}\,.
\end{equation}

\item
Next we determine $\varphi$. Once more using (\ref{eq:EquivBoostsRot}),
where $\mat{D}$ is a $\pi$-rotation about the $y$-axis, shows that 
the functions $a$ and $d$ in (\ref{eq:FormBoost1}) are even and the 
functions $b$ and $c$ are odd. The definition (\ref{eq:FormBoost4b}) 
of $\varphi$ then implies that $\varphi$ is odd. Since we assumed 
$\vec v$ to be a continuous coordinatisation of a topological group,
the map $\varphi$ must also be continuous (since the inversion 
map, $g\mapsto g^{-1}$, is continuous in a topological group). 
A standard theorem now states that a continuous bijection of an 
interval of $\reals$ onto itself must be strictly monotonic. 
Together with (\ref{eq:FormBoost6}) this implies that $\varphi$ is 
either the identity or minus the identity map.\footnote{The simple 
proof is as follows, where we write $v':=\varphi(v)$ to save notation, 
so that (\ref{eq:FormBoost6}) now reads $v''=v$. First assume that 
$\varphi$ is strictly monotonically increasing, then $v'>v$ implies 
$v=v''>v'$, a contradiction, and $v'<v$ implies $v=v''<v'$, likewise 
a contradiction. Hence $\varphi=\text{id}$ in this case. Next assume 
$\varphi$ is strictly monotonically decreasing. Then 
$\tilde\varphi:=-\varphi$ is a strictly monotonically increasing map 
of the interval $(-c,c)$ to itself that obeys (\ref{eq:FormBoost6}). 
Hence, as just seen, $\tilde\varphi=\text{id}$, i.e. $\varphi=-\text{id}$.} If it is the 
identity map, evaluation of (\ref{eq:FormBoost5}) shows that either 
the determinant of $\mat{A}(v)$ must equals $-1$, or that $\mat{A}(v)$ 
is the identity for all $\vec v$. We exclude the second possibility 
straightaway and the first one on the grounds that we required 
$\mat{A}(v)$ be the identity for $v=0$. Also, in that case, 
(\ref{eq:FormBoost5}) implies $A^2(v)=\text{id}$ for all $v\in(-c,c)$. 
We conclude that $\varphi=-\text{id}$, which implies that the relative 
velocity of $K$ with respect to $K'$ is minus the relative velocity 
of $K'$ with respect to $K$. Plausible as it might seem, there is 
no a priori reason why this should be so.\footnote{Note that $v$ and $v'$ 
are measured with different sets of rods and clocks.}. On the face of it,
the RP only implies (\ref{eq:FormBoost6}), not the stronger relation 
$\varphi(v)=-v$. This was first pointed out in~\cite{Berzi.Gorini:1969}. 

\item
We briefly revisit (\ref{eq:FormBoosts2}). Since we have seen 
that $B(-v\vec e_x)$ is the inverse of $B(v\vec e_x)$, we must have   
$\alpha(-v)=1/\alpha(v)$, so that (\ref{eq:FormBoosts2})
implies $\alpha(v)\equiv\pm 1$. But only $\alpha(v)\equiv +1$ is 
compatible with our requirement that $B(\vec 0)$ be the identity. 

\item   
Now we return to the determination of $\mat{A}(v)$. 
Using (\ref{eq:FormBoost4}) and $\varphi=-\text{id}$, we write 
\begin{equation}
\label{eq:FormBoost6b}
\mat{A}(v)=
\begin{pmatrix}
a(v)&b(v)\\
-va(v)&a(v)
\end{pmatrix}
\end{equation}
and 
\begin{equation}
\label{eq:FormBoost7}
\Delta(v):=\det\bigl(\mat{A}(v)\bigr)
=a(v)\bigl[a(v)+vb(v)\bigr]\,.
\end{equation}
Equation $\mat{A}(-v)=(\mat{A}(v))^{-1}$ is now equivalent to 
\begin{subequations}
\label{eq:FormBoost8}
\begin{alignat}{2}
\label{eq:FormBoost8a}
& a(-v)&&\,=\,a(v)/\Delta(v)\,,\\
\label{eq:FormBoost8b}
& b(-v)&&\,=\,-\,b(v)/\Delta(v)\,.
\end{alignat}
\end{subequations}
Since, as already seen, $a$ is an even and $b$ is an odd function,
(\ref{eq:FormBoost8}) is equivalent to $\Delta(v)\equiv 1$, i.e. 
the unimodularity of $B(\vec v)$. Equation (\ref{eq:FormBoost7})
then allows to express $b$ in terms of $a$:
\begin{equation}
\label{eq:FormBoost9}
b(v)=\frac{a(v)}{v}\left[\frac{1}{a^2(v)}-1\right]\,.
\end{equation}

\item
Our problem is now reduced to the determination of the single 
function $a$. This we achieve by employing the requirement that 
the composition of two boosts in the same direction results again 
in a boost in that direction, i.e. 
\begin{equation}
\label{eq:FormBoost10}
\mat{A}(v)\cdot\mat{A}(v')=\mat{A}(v'')\,.
\end{equation}
According to (\ref{eq:FormBoost6b}) each matrix $\mat{A}(v)$ has equal 
diagonal entries. Applied to the product matrix  on the left hand side 
of (\ref{eq:FormBoost10}) this implies that $v^{-2}(a^{-2}(v)-1)$ is 
independent of $v$, i.e. equal to some constant $k$ whose physical 
dimension is that of an inverse velocity squared. Hence we have
\begin{equation}
\label{eq:FormBoost11}
a(v)=\frac{1}{\sqrt{1+kv^2}}\,,
\end{equation}
where we have chosen the positive square root since we require 
$a(0)=1$. The other implications of (\ref{eq:FormBoost10}) are 
\begin{subequations}
\label{eq:FormBoost12}
\begin{alignat}{2}
\label{eq:FormBoost12a}
& a(v)a(v')(1-kvv')&&\,=\,a(v'')\,,\\
\label{eq:FormBoost12b}
& a(v)a(v')(1+vv')&&\,=\,v''a(v'')\,,
\end{alignat}
\end{subequations}
from which we deduce 
\begin{equation}
\label{eq:FormBoost13}
v''=\frac{v+v'}{1-kvv'}\,.
\end{equation}
Conversely, (\ref{eq:FormBoost11}) and (\ref{eq:FormBoost13})
imply (\ref{eq:FormBoost12}). We conclude that 
(\ref{eq:FormBoost10}) is equivalent to (\ref{eq:FormBoost11}) 
and (\ref{eq:FormBoost13}).  

\item
So far a boost in $x$ direction has been shown to act non-trivially 
only in the $t-x$ plane, where its action is given by the matrix that 
results from inserting (\ref{eq:FormBoost9}) and (\ref{eq:FormBoost11})
into (\ref{eq:FormBoost6b}): 
\begin{equation}
\label{eq:FormBoost14}
\mat{A}(v)=\begin{pmatrix}
a(v)&kv\,a(v)\\
-v\,a(v)& a(v)\,,
\end{pmatrix}
\qquad\text{where}\quad
a(v)=1/\sqrt{1+kv^2}\,.
\end{equation}
\begin{itemize}
\item
If $k>0$ we rescale $t\mapsto \tau:= t/\sqrt{k}$ and set 
$\sqrt{k}\,v:=\tan\alpha$. Then (\ref{eq:FormBoost14}) is seen to 
be a Euclidean rotation with angle $\alpha$ in the $\tau-x$ 
plane. The velocity spectrum is the whole real line plus 
infinity, i.e. a circle, corresponding to $\alpha\in[0,2\pi]$, 
where $0$ and $2\pi$ are identified. Accordingly, the composition 
law (\ref{eq:FormBoost13}) is just ordinary addition for the angle 
$\alpha$. This causes several paradoxa when $v$ is interpreted 
as velocity. For example, composing two finite 
velocities $v,v'$ which satisfy $vv'=1/k$ results in  
$v''=\infty$, and composing two finite and positive velocities, 
each of which is greater than $1/\sqrt{k}$, results in a finite 
but negative velocity. In this way the successive composition 
of finite positive velocities could also result in zero velocity. 
The group $\group{G}\subset\group{GL}(n,\reals)$ obtained in 
this fashion is, in fact, $\group{SO}(4)$. This group may be 
uniquely characterised as the largest connected group of bijections 
of $\reals^4$ that preserves the Euclidean distance measure. 
In particular, it treats time symmetrically with all space directions,
so that no invariant notion of time-orientability can be given in 
this case.  
\item
For $k=0$ the transformations are just the ordinary boosts of 
the Galilei group. The velocity spectrum is the whole real 
line (i.e. $v$ is unbounded but finite) and $\group{G}$ is the 
Galilei group. The law for composing velocities is just 
ordinary vector addition. 
\item
Finally, for $k<0$, one infers from (\ref{eq:FormBoost13})
that $c:=1/\sqrt{-k}$ is an upper bound for all velocities, 
in the sense that composing two velocities taken from the
interval $(-c,c)$ always results in a velocity from within 
that interval. Writing $\tau:=ct$, $v/c=:\beta=:\tanh\rho$, 
and $\gamma=1/\sqrt{1-\beta^2}$, the matrix (\ref{eq:FormBoost14}) 
is seen to be a \emph{Lorentz boost} or \emph{hyperbolic motion} in the 
$\tau-x$ plane:
\begin{equation}
\label{eq:FormBoost15}
\begin{pmatrix}\tau\\ x\end{pmatrix}
\mapsto
\begin{pmatrix}
\gamma&-\beta\gamma\\
-\beta\gamma&\gamma
\end{pmatrix}\cdot
\begin{pmatrix}\tau\\ x\end{pmatrix}=
\begin{pmatrix}
\cosh\rho&-\sinh\rho\\
-\sinh\rho&\cosh\rho
\end{pmatrix}\cdot
\begin{pmatrix}\tau\\ x\end{pmatrix}\,.
\end{equation}
The quantity 
\begin{equation}
\label{eq:def:Rapidity}
\rho:=\tanh^{-1}(v/c)=\tanh^{-1}(\beta)
\end{equation}
is called \emph{rapidity}\footnote{This term was coined by 
Robb~\cite{Robb:OptGeom}, but the quantity was used before by 
others; compare \cite{Varicak:1912}.}. If rewritten in terms of 
the corresponding rapidities the composition law 
(\ref{eq:FormBoost13}) reduces to ordinary addition: 
$\rho''=\rho+\rho'$.  
\end{itemize} 
\end{enumerate}

This shows that only the Galilei and the Lorentz group 
survive as candidates for any symmetry group implementing the 
RP. Once the Lorentz group for velocity parameter $c$ is 
chosen, one may fully characterise it by its property to 
leave a certain symmetric bilinear form invariant.
In this sense we geometric structure of Minkowski space can 
be deduced. This closes the circle to where we started from 
in Section\,\ref{sec:FromMetricToAffine}.

\subsection{Local versions}
\label{sec:LocallyAffineMaps}
In the previous sections we always understood an automorphisms of a 
structured set (spacetime) as a bijection. Mathematically this seems
an obvious requirement, but from a physical point of view this is 
less clear. The physical law of inertia provides us with distinguished 
motions \emph{locally} in space and time. Hence one may attempt to relax 
the condition for structure preserving maps, so as to only preserve 
inertial motions \emph{locally}. Hence we ask the following question:
What are the most general maps that \emph{locally} map segments of 
straight lines to segments of straight lines? This local approach has 
been pursued by~\cite{Fock:STG}.  

To answer this question completely, let us (locally) identify spacetime 
with $\reals^n$ where $n\geq 2$ and assume the map to be $C^3$, that is, 
three times continuously differentiable.\footnote{This requirement 
distinguishes the present (local) from the previous (global) approaches, 
in which not even continuity needed to be assumed.} So let 
$U\subseteq\reals^n$ be an open subset and determine all $C^3$ maps 
$f:U\rightarrow\reals^n$ that map straight segments in $U$ into straight 
segments in $\reals^n$. In coordinates we write $x=(x^1,\cdots,x^n)\in U$ 
and $y=(y^1,\cdots,y^n)\in f(U)\subseteq\reals^n$, so that 
$y^\mu:=f^\mu(x)$. A straight segment in $U$ is a curve 
$\gamma:I\rightarrow U$ (the open interval $I\subseteq\reals$
is usually taken to contain zero) whose acceleration is pointwise 
proportional to its velocity. This is equivalent to saying that 
it can be parametrised so as to have zero acceleration, i.e., 
$\gamma(s)=as+b$ for some $a,b\in\reals^n$. 

For the image path $f\circ\gamma$ to be again straight its 
acceleration, $(f''\circ\gamma)(a,a)$, must be proportional to
its velocity, $(f'\circ\gamma)(a)$, where the factor of 
proportionality, $C$, depends on the point of the path and 
separately on $a$. Hence, in coordinates, we have 
\begin{equation}
\label{eq:ProjMap1}
f^\mu_{,\lambda\sigma}(as+b)a^\lambda a^\sigma
=f^\mu_{,\nu}(as+b)a^\nu\,C(as+b,a)
\end{equation}
For each $b$ this must be valid for all $(a,s)$ in a neighbourhood 
of zero in $\reals^n\times\reals$. Taking the second derivatives 
with respect to $a$, evaluation at $a=0$, $s=0$ leads to 
\begin{subequations}
\label{eq:ProjMap2}
\begin{equation}
\label{eq:ProjMap2-a}
f^\mu_{,\lambda\sigma}=\Gamma^\nu_{\lambda\sigma}f^\mu_{,\nu}\,,
\end{equation}
where
\begin{alignat}{2}
\label{eq:ProjMap2-b}
&\Gamma^\nu_{\lambda\sigma}&&\,:=\,
\delta^\nu_\lambda\psi_\sigma+
\delta^\nu_\sigma\psi_\lambda\\
\label{eq:ProjMap2-c}
&\psi_\sigma &&\,:=\,
\frac{\partial C(\cdot,a)}{\partial a^\sigma}\bigg\vert_{a=0}
\end{alignat}
\end{subequations}
Here we suppressed the remaining argument $b$. Equation (\ref{eq:ProjMap2})
is valid at each point in $U$. Integrability of (\ref{eq:ProjMap2-a})
requires that its further differentiation is totally symmetric with 
respect to all lower indices (here we use that the map $f$ is $C^3$).
This leads to 
\begin{equation}
\label{eq:ProjMap3}
R^\mu_{\phantom{\mu}\alpha\beta\gamma}:=
\partial_\beta\Gamma_{\alpha\gamma}^\mu+
\Gamma^\nu_{\sigma\beta}\Gamma^\sigma_{\alpha\gamma}
\,-\,(\beta\leftrightarrow\gamma)\,=\,0\,.
\end{equation}
Inserting (\ref{eq:ProjMap2-b}) one can show (upon taking traces over 
$\mu\alpha$ and $\mu\gamma$)  that the resulting 
equation is equivalent to 
\begin{equation}
\label{eq:ProjMap4}
\psi_{\alpha,\beta}=\psi_\alpha\psi_\beta\,.
\end{equation}
In particular $\psi_{\alpha,\beta}=\psi_{\beta,\alpha}$ so that there 
is a local function $\psi:U\rightarrow\reals$ (if $U$ is simply connected, 
as we shall assume) for which $\psi_\alpha=\psi_{,\alpha}$. 
Equation (\ref{eq:ProjMap4}) is then 
equivalent to $\partial_\alpha\partial_\beta\exp(-\psi)=0$ so that 
$\psi(x)=-\ln(p\cdot x+q)$ for some $p\in\reals^n$ and $q\in\reals$. 
Using $\psi_\sigma=\psi_{,\sigma}$ and (\ref{eq:ProjMap4}), equation 
(\ref{eq:ProjMap2-a}) is equivalent to 
$\partial_\lambda\partial_\sigma\bigl[f^\mu\exp(-\psi)\bigr]=0$,
which finally leads to the result that the most general solution for 
$f$ is given by  
\begin{equation}
\label{eq:ProjMap5}
f(x)=\frac{A\cdot x+a}{p\cdot x+q}\,.
\end{equation}
Here $A$ is a $n\times n$ matrix, $a$ and $q$ vectors in $\reals^n$,
and $q\in\reals$. $p$ and $q$ must be such that $U$ does not intersect 
the hyperplane $H(p,q):=\{x\in\reals^n\mid p\cdot x+q=0\}$ where $f$ 
becomes singular, but otherwise they are arbitrary. 
Iff $H(p,q)\ne\emptyset$, i.e. iff $p\ne 0$, the transformations 
(\ref{eq:ProjMap5}) are not affine. In this case they are called 
proper projective.  

Are there physical reasons to rule out such proper projective 
transformations? A structural argument is that they do not 
leave any subset of $\reals^n$ invariant and that they hence 
cannot be considered as automorphism group of any subdomain.   
A physical argument is that two separate points that move 
with the same velocity cease to do so if their worldlines are 
transformed by by a proper projective transformation. In particular, 
a rigid motion of an extended body (undergoing inertial motion) 
ceases to be rigid if so transformed 
(cf.\cite{Dixon:SpecialRelativity}, p.\,16). An illustrative 
example is the following: Consider the one-parameter 
($\sigma$) family of parallel lines $x(s,\sigma)=se_0+\sigma e_1$ 
(where $s$ is the parameter along each line), and the proper 
projective map $f(x)=x/(-e_0\cdot x+1)$ which becomes singular 
on the hyperplane $x^0=1$. The one-parameter family of image 
lines 
\begin{equation}
\label{eq:ImageLines}
y(s,\sigma):=f\bigl(x(s,\sigma)\bigr)
=\frac{se_0+\sigma e_1}{1-s}
\end{equation}
have velocities 
\begin{equation}
\label{eq:ImageLinesVel}
\partial_sy(s,\sigma)=\frac{qe_0+\sigma e_1}{(1-s)^2}
\end{equation}
whose directions are independent of $s$, showing that they are 
indeed straight. However, the velocity directions  now depend 
on $\sigma$, showing that they are not parallel anymore. 
 
Let us, regardless of this, for the moment take seriously the 
transformations (\ref{eq:ProjMap5}). One may reduce them to 
the following form of generalised boosts, discarding translations 
and rotations and using equivariance with respect to the latter 
(we restrict to four spacetime dimensions from now on): 
\begin{subequations}
\label{eq:FockLorentzTrans1}
\begin{alignat}{2}
\label{eq:FockLorentzTrans1-a}
& t'&&\,=\,
\frac{a(v)t+b(v)\vec(v\cdot\vec x)}{A(v)+B(v)t+D(v)(\vec v\cdot\vec x)}\,,\\
\label{eq:FockLorentzTrans1-b}
&\vec x'_{\Vert}&&\,=\,
\frac{d(v)\vec vt+e(v)\vec x_{\Vert}}{A(v)+B(v)t+D(v)(\vec v\cdot\vec x)}\,,\\
\label{eq:FockLorentzTrans1-c}
&\vec x'_{\perp}&&\,=\,
\frac{f(v)\vec x_{\perp}}{A(v)+B(v)t+D(v)(\vec v\cdot\vec x)}\,.
\end{alignat}
\end{subequations}
where $\vec v\in\reals^3$ represents the boost velocity, $v:=\Vert\vec v\Vert$
its modulus, and all functions of $v$ are even. The subscripts $\Vert$ and 
$\perp$ refer to the components parallel and perpendicular to $\vec v$. 
Now one imposes the following conditions which allow to determine the  
eight functions $a,b,d,e,f,A,B,D$, of which only seven are considered 
independent since common factors of the numerator and denominator cancel
(we essentially follow~\cite{Manida:1999}):
\begin{enumerate}
\item 
The origin $\vec x'=0$ has velocity $\vec v$ in the unprimed 
coordinates, leading to $e(v)=-d(v)$ and thereby eliminating 
$e$ as independent function. 
\item
The origin $\vec x=0$ has velocity $-\vec v$ in the primed 
coordinates, leading to $d(v)=-a(v)$ and thereby eliminating 
$d$ as independent function.
\item
Reciprocity: The transformation parametrised by $-\vec v$ is the 
inverse of that parametrised by $\vec v$, leading to relations
$A=A(a,b,v)$, $B=B(D,a,b,v)$, and $f=A$, thereby eliminating 
$A,B,f$ as independent functions. Of the remaining three 
functions $a,b,D$ an overall factor in the numerator and 
denominator can be split off so that two free functions remain. 
\item
Transitivity: The composition of two transformations of the 
type (\ref{eq:FockLorentzTrans1}) with parameters $\vec v$
and $\vec v'$ must be again of this form with some 
parameter $\vec v''(\vec v,\vec v')$, which turns out to be 
the same function of the velocities $\vec v$ and $\vec v'$ as 
in Special Relativity (Einstein's addition law), for reasons to 
become clear soon. This allows to determine the last two 
functions in terms of two constants $c$ and $R$ whose physical 
dimensions are that of a velocity and of a length respectively. 
Writing, as usual, $\gamma(v):=1/\sqrt{1-v^2/c^2}$ the final 
form is given by 
\end{enumerate}
\begin{subequations}
\label{eq:FockLorentzTrans2}
\begin{alignat}{2}
\label{eq:FockLorentzTrans2-a}
& t'&&\,=\,
\frac{\gamma(v)(t-\vec v\cdot\vec x/c^2)}%
{1-\bigl(\gamma(v)-1\bigr)ct/R+\gamma(v)\vec v\cdot\vec x/Rc}\,,\\
\label{eq:FockLorentzTrans2-b}
&\vec x'_{\Vert}&&\,=\,
\frac{\gamma(v)(\vec x_\Vert-\vec v t)}%
{1-\bigl(\gamma(v)-1\bigr)ct/R+\gamma(v)\vec v\cdot\vec x/Rc}\,,\\
\label{eq:FockLorentzTrans2-c}
&\vec x'_{\perp}&&\,=\,
\frac{\vec x_\perp}%
{1-\bigl(\gamma(v)-1\bigr)ct/R+\gamma(v)\vec v\cdot\vec x/Rc}\,.
\end{alignat}
\end{subequations}
In the limit as $R\rightarrow\infty$ this approaches an ordinary 
Lorentz boost:
\begin{equation}
\label{eq:LorentzFock3}
L(\vec v): (t,\vec x_\Vert,\vec x_\perp)\mapsto 
\bigl(\gamma(v)(t-\vec v\cdot\vec x/c^2)\,,\,
\gamma(v)(\vec x_\Vert-\vec vt)\,,\,
\vec x_\perp\bigr)\,.
\end{equation}
Moreover, for finite $R$ the map (\ref{eq:FockLorentzTrans2}) is 
conjugate to (\ref{eq:LorentzFock3}) with respect to a time 
dependent deformation. To see this, observe that the common 
denominator in (\ref{eq:FockLorentzTrans2}) is just $(R+ct)/(R+ct')$, 
whereas the numerators correspond to (\ref{eq:LorentzFock3}). Hence, 
introducing the deformation map 
\begin{equation}
\label{eq:FockLorentz3}
\phi: (t,\vec x)\mapsto 
\left(\frac{t}{1-ct/R}\,,\,\frac{\vec x}{1-ct/R}\right)
\end{equation}
and denoting the map $(t,\vec x)\mapsto (t',\vec x')$ in 
(\ref{eq:FockLorentzTrans2}) by $f$, 
we have 
\begin{equation}
\label{eq:FockLorentz4}
f=\phi\circ L(\vec v)\circ\phi^{-1}\,.
\end{equation}
Note that $\phi$ is singular at the hyperplane $t=R/c$ and has no 
point of the hyperplane $t=-R/c$ in its image. The latter hyperplane 
is the singularity set of $\phi^{-1}$. Outside the hyperplanes 
$t=\pm R/c$ the map $\phi$ relates the following time slabs in a 
diffeomorphic fashion:
\begin{subequations}
\label{eq:DiffTimeSlabs}
\begin{alignat}{11}
\label{eq:DiffTimeSlabs-a}
  &0      &&\,\leq\,  && \ t\,&&\,<\, &&R/c    
&&\quad\mapsto\quad  &&0       &&\,\leq\, &&\ t\,&&\,<\, &&\infty\,,\\
\label{eq:DiffTimeSlabs-b}
  &R/c    &&\, <\, && \ t\,&&\, <\, &&\infty 
&&\quad\mapsto\quad -&&\infty  &&\,<\, &&\ t\,&&\,<\, -&&R/c\,,\\
\label{eq:DiffTimeSlabs-c}
 -&\infty &&\,<\,&& \ t\,&&\leq &&0      
&&\quad\mapsto\quad -&&R/c  &&\,<\,&&\ t\, &&\,\leq\,&&0\,.
\end{alignat}
\end{subequations}
Since boosts leave the upper-half spacetime, $t>0$, invariant (as set), 
(\ref{eq:DiffTimeSlabs-a}) shows that $f$ just squashes the linear 
action of boosts in  $0<t<\infty$ into a non-linear action within 
$0<t<R/c$, where $R$ now corresponds to an invariant scale. Interestingly, 
this is the same deformation of boosts that have been recently considered 
in what is sometimes called \emph{Doubly Special Relativity} (because there 
are now two, rather than just one, invariant scales, $R$ and $c$), 
albeit there the deformation of boosts take place in momentum space 
where $R$ then corresponds to an invariant energy scale; see 
\cite{Magueijo.Smolin:2002} and also \cite{Jafari.Shariati:2004}.

\section{Selected structures in Minkowski space}
\label{sec:SelectedStructures}
In this section we wish to discuss in more detail some 
of the non-trivial structures in Minkowski. I have chosen 
them so as to emphasise the difference to the corresponding 
structures in Galilean spacetime, and also because they 
do not seem to be much discussed in other standard sources.

\subsection{Simultaneity} 
\label{sec:SelectedStructuresSimultaneity}
Let us start right away by characterising those vectors 
for which we have an inverted Cauchy-Schwarz inequality:  
\begin{lemma}
\label{lemma:StrictICSI}
Let $V$ be of dimension $n>2$ and $v\in V$ be some 
non-zero vector. The strict inverted Cauchy-Schwarz inequality,
\begin{equation}
\label{eq:StrictInvCSI}
v^2w^2<(v\cdot w)^2\,,
\end{equation}
holds for all $w\in V$ linearly independent of $v$ iff $v$ is timelike.  
\end{lemma}
\begin{proof}
Obviously $v$ cannot be spacelike, for then we would violate 
(\ref{eq:StrictInvCSI}) with any spacelike $w$. If $v$ is 
lightlike then $w$ violates (\ref{eq:StrictInvCSI}) iff it is 
in the set $v^\perp-\Span\{v\}$, which is non-empty iff 
$n>2$. Hence $v$ cannot be lightlike if $n>2$. If $v$ is timelike
we decompose $w=av+w'$ with $w'\in v^\perp$ so that $w'^2\leq 0$,
with equality iff $v$ and $w$ are linearly dependent. Hence 
\begin{equation}  
\label{eq:LemmaStrictICSI1}
(v\cdot w)^2-v^2w^2=-v^2\,w'^2\geq 0\,,
\end{equation}
with equality iff $v$ and $w$ are linearly dependent. 
\end{proof}
 
\noindent
The next Lemma deals with the intersection of a causal line
with a light cone, a situation depicted in Fig.\,\ref{fig:Robb}.

\begin{lemma}
\label{lemma:IntLineCone}
Let $\mathcal{L}_p$ be the light-doublecone with vertex $p$ and 
$\ell:=\{r+\lambda v\mid r\in\reals\}$ be a non-spacelike line, 
i.e. $v^2\geq 0$, through $r\not\in\mathcal{L}_p$. If $v$ is timelike 
$\ell\cap\mathcal{L}_p$ consists of two points. If $v$ is lightlike 
this intersection consists of one point if $p-r\not\in v^\perp$
and is empty if $p-r\in v^\perp$. Note that the latter two statements 
are independent of the choice of $r\in\ell$---as they must be---, i.e. 
are invariant under $r\mapsto r':=r+\sigma v$, where $\sigma\in\reals$.
\end{lemma}
\begin{proof}
We have $r+\lambda v\in\mathcal{L}_p$ iff
\begin{equation}
\label{eq:LemmaIntLineCone1}
(r+\lambda v-p)^2=0\ \Longleftrightarrow\
\lambda^2 v^2+2\lambda v\cdot(r-p)+(r-p)^2=0\,.
\end{equation}
For $v$ timelike we have $v^2>0$ and (\ref{eq:LemmaIntLineCone1})
has two solutions 
\begin{equation}
\label{eq:LemmaIntLineCone2}
\lambda_{1{,}2}=\frac{1}{v^2}\left\{
-v\cdot (r-p)\pm\sqrt{\bigl(v\cdot(r-p)\bigr)^2-v^2(r-p)^2}\,\right\}\,.
\end{equation}
Indeed, since $r\not\in\mathcal{L}_p$, the vectors $v$ and $r-p$ 
cannot be linearly dependent so that Lemma\,\ref{lemma:StrictICSI} 
implies the positivity of the expression under the square root. 
If $v$ is lightlike (\ref{eq:LemmaIntLineCone1}) becomes a linear 
equation which is has one solution if $v\cdot(r-p)\ne 0$ and no solution if 
$v\cdot (r-p)=0$ [note that $(r-p)^2\ne 0$ since $q\not\in \mathcal{L}_p$
by hypothesis].
\end{proof}

\begin{figure}[h!]
\centering\epsfig{figure=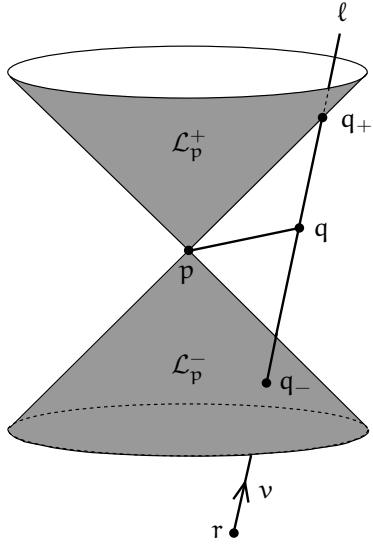,width=0.38\linewidth}
\put(-12,155){\small $q_+$}
\put(-22,114){\small $q$}
\put(-35,56){\small $q_-$}
\put(-72,98){\small $p$}
\put(-60,0){\small $r$}
\put(-42,15){\small $v$}
\put(-75,146){\small $\mathcal{L}^+_p$}
\put(-75,62){\small $\mathcal{L}^-_p$}
\put(-13,195){\small $\ell$}
\caption{\label{fig:Robb}
A timelike line $\ell=\{r+\lambda v\mid\lambda\in\reals\}$ intersects 
the light-cone with vertex $p\not\in\ell$ in two points: $q_+$, its  
intersection with the future light-cone and $q_-$, its intersection 
with past the light cone. $q$ is a point in between $q_+$ and $q_-$.}
\end{figure}

\begin{proposition}
\label{prop:EinsteinSynch}
Let $\ell$ and $\mathcal{L}_p$ as in Lemma\,\ref{lemma:IntLineCone} with $v$ 
timelike. Let $q_+$ and $q_-$ be the two intersection points of 
$\ell$ with $\mathcal{L}_p$ and $q\in\ell$ a point between them. Then 
\begin{equation}
\label{eq:prop:EinsteinSynch1}
\Vert q-p\Vert_g^2=\Vert q_+-q\Vert_g\,\Vert q-q_-\Vert_g\,.
\end{equation}
Moreover, $\Vert q_+-q\Vert_g=\Vert q-q_-\Vert_g$ iff 
$p-q$ is perpendicular to $v$. 
\end{proposition}
\begin{proof}
The vectors $(q_+-p)=(q-p)+(q_+-q)$ and $(q_--p)=(q-p)+(q_--q)$ 
are lightlike, which gives (note that $q-p$ is spacelike):
\begin{subequations}
\label{eq:prop:EinsteinSynch23}
\begin{alignat}{3}
\label{eq:prop:EinsteinSynch2}
& \Vert q-p\Vert_g^2 
&&\,=\,-(q-p)^2 
&&\,=\,(q_+-q)^2+2(q-p)\cdot(q_+-q)\,,\\
\label{eq:prop:EinsteinSynch3}
& \Vert q-p\Vert_g^2 
&&\,=\,-(q-p)^2
&&\, =\,(q_--q)^2+2(q-p)\cdot(q_--q)\,.
\end{alignat}
\end{subequations}
Since $q_+-q$ and $q-q_-$ are parallel we have 
$q_+-q=\lambda (q-q_-)$ with $\lambda\in\reals_+$ so that 
$(q_+-q)^2=\lambda\Vert q_+-q\Vert_g\Vert q-q_-\Vert_g$ and
$\lambda (q_--q)^2=\Vert q_+-q\Vert_g\Vert q-q_-\Vert_g$.
Now, multiplying (\ref{eq:prop:EinsteinSynch3}) with $\lambda$ 
and adding this to (\ref{eq:prop:EinsteinSynch2}) immediately 
yields 
\begin{equation}
\label{eq:prop:EinsteinSynch4}
(1+\lambda)\,\Vert q-p\Vert_g^2
=(1+\lambda)\,\,\Vert q_+-q\Vert_g\Vert q-q_-\Vert_g\,.\\
\end{equation}
Since $1+\lambda\ne 0$ this implies (\ref{eq:prop:EinsteinSynch1}).
Finally, since $q_+-q$ and $q_--q$ are antiparallel, 
$\Vert q_+-q\Vert_g=\Vert q_--q\Vert_g$ iff $(q_+-q)=-(q_--q)$. 
Equations (\ref{eq:prop:EinsteinSynch23}) now show that this is the 
case iff $(q-p)\cdot(q_\pm-q)=0$, i.e. iff $(q-p)\cdot v=0$. 
Hence we have shown  

\begin{equation}
\label{eq:prop:EinsteinSynch5}
\Vert q_+-q\Vert_g=\Vert q-q_-\Vert_g\ \Longleftrightarrow\ 
(q-p)\cdot v=0\,.
\end{equation}
In other words, $q$ is the midpoint of the segment $\overline{q_+q_-}$ 
iff the line through $p$ and $q$ is perpendicular (wrt. $g$) to $\ell$.
\end{proof}
The somewhat surprising feature of the first statement of this 
proposition is that (\ref{eq:prop:EinsteinSynch1}) holds for 
\emph{any} point of the segment $\overline{q_+q_-}$, not just 
the midpoint, as it would have to be the case for the corresponding 
statement in Euclidean geometry.

The second statement of Proposition\,\ref{prop:EinsteinSynch} gives 
a convenient geometric characterisation of Einstein-simultaneity.
Recall that an event $q$ on a timelike line $\ell$ (representing an 
inertial observer) is defined to be Einstein-simultaneous with an 
event $p$ in spacetime iff $q$ bisects the segment $\overline{q_+q_-}$ 
between the intersection points $q_+,q_-$ of $\ell$ with the 
double-lightcone at $p$. Hence Proposition\,\ref{prop:EinsteinSynch} 
implies
\begin{corollary}
\label{corol:EinsteinSim}
Einstein simultaneity with respect to a timelike line $\ell$ is 
an equivalence relation on spacetime, the equivalence classes of 
which are the spacelike hyperplanes orthogonal (wrt. $g$) to $\ell$. 
\end{corollary}
The first statement simply follows from the fact that the family 
of parallel hyperplanes orthogonal to $\ell$ form a partition 
(cf. Sect.\,\ref{sec:GroupActions}) of spacetime. 

From now on we shall use the terms `timelike line' and `inertial 
observer' synonymously. Note that Einstein simultaneity is only defined 
relative to an inertial observer. Given two inertial observers, 
\begin{subequations}
\label{eq:TwoObservers}
\begin{alignat}{3}
\label{eq:TwoObservers1}
&\ell&&\,=\,\{r+\lambda v\mid\lambda\in\reals\}\qquad
&&\text{first observer}\,,\\
\label{eq:TwoObservers2}
&\ell'&&\,=\,\{r'+\lambda' v'\mid\lambda'\in\reals\}\qquad
&&\text{second observer}\,,
\end{alignat}
\end{subequations}
we call the corresponding Einstein-simultaneity relations 
\emph{$\ell$-simultaneity} and \emph{$\ell'$-simultaneity}. Obviously 
they coincide iff $\ell$ and $\ell'$ are parallel ($v$ and $v'$ 
are linearly dependent). 
In this case $q'\in\ell'$ is $\ell$-simultaneous to $q\in\ell$ iff
$q\in\ell$ is $\ell'$-simultaneous to $q'\in\ell'$. If $\ell$ and 
$\ell'$are not parallel (skew or intersecting in one point) it is 
generally not true that if $q'\in\ell'$ is $\ell$-simultaneous to 
$q\in\ell$ then $q\in\ell$ is also $\ell'$-simultaneous to $q'\in\ell'$. 
In fact, we have
\begin{proposition}
\label{prop:SkewSim}
Let $\ell$ and $\ell'$ two non-parallel timelike likes. 
There exists a unique pair $(q,q')\in\ell\times\ell'$ so that $q'$ is 
$\ell$-simultaneous to $q$ and $q$ is $\ell'$ simultaneous to $q'$.
\end{proposition}
\begin{proof}
We parameterise $\ell$ and $\ell'$ as in  (\ref{eq:TwoObservers}).  
The two conditions for $q'$ being $\ell$-simultaneous to $q$ and $q$ 
being $\ell'$-simultaneous to $q'$ are $(q-q')\cdot v=0=(q-q')\cdot v'$.
Writing $q=r+\lambda v$ and $q'=r'+\lambda'v'$ this takes the form 
of the following matrix equation for the two unknowns $\lambda$ and 
$\lambda'$: 
\begin{equation}
\label{eq:prop:SkewSim1}
\begin{pmatrix}%
v^2&-v\cdot v'\\
v\cdot v'&-v'^2
\end{pmatrix}
\begin{pmatrix}
\lambda\\
\lambda'
\end{pmatrix}=
\begin{pmatrix}
(r'-r)\cdot v\\
(r'-r)\cdot v'
\end{pmatrix}\,.
\end{equation}
This has a unique solution pair $(\lambda,\lambda')$, since for 
linearly independent timelike vectors $v$ and $v'$  
Lemma\,\ref{lemma:StrictICSI} implies $(v\cdot v')^2-v^2v'^2>0$.
Note that if $\ell$ and $\ell'$ intersect
$q=q'=\text{intersection point}$.
\end{proof} 

Clearly, Einstein-simultaneity is conventional and physics proper 
should not depend on it. For example, the fringe-shift in the 
Michelson-Morley experiment is independent of how we choose to 
synchronise clocks. In fact, it does not even make use of any clock. 
So what is the general definition of a `simultaneity structure'?
It seems obvious that it should be a relation on spacetime that is 
at least symmetric (each event should be simultaneous to itself). 
Going from one-way simultaneity to the mutual synchronisation of 
two clocks, one might like to also require reflexivity (if $p$ is 
simultaneous to $q$ then $q$ is simultaneous to $p$), though this 
is not strictly required in order to one-way synchronise each clock 
in a set of clocks with one preferred `master clock', which is 
sufficient for many applications. 

Moreover, if we like to speak of the mutual simultaneity of sets of 
more than two events we need an equivalence relation on spacetime. 
The equivalence relation should be such that each inertial observer 
intersect each equivalence class precisely once. Let us call such 
a simultaneity structure `admissible'. Clearly there are zillions of 
such structures: just partition spacetime into any set of 
appropriate\footnote{For example, the hypersurfaces should not be 
asymptotically hyperboloidal, for then a constantly accelerated observer 
would not intersect all of them.} spacelike 
hypersurfaces (there are more possibilities at this point, like 
families of forward or backward lightcones). An \emph{absolute} 
admissible simultaneity structure would be one which is invariant 
(cf. Sect.\,\ref{sec:GroupActions}) under the automorphism group of 
spacetime. We have 
\begin{proposition}
\label{prop:AbsSim}
There exits precisely one admissible simultaneity 
structure which is invariant under the inhomogeneous proper 
orthochronous Galilei group and none that is invariant under the 
inhomogeneous proper orthochronous Lorentz group.
\end{proposition}
A proof is given in \cite{Giulini:2002a}. There is a group-theoretic 
reason that highlights this existential difference: 
\begin{proposition}
\label{prop:AbsSimMath} 
Let $G$ be a group with transitive action on a set $S$. Let 
$\Stab(p)\subset G$ be the stabiliser subgroup for $p\in S$
(due to transitivity all stabiliser subgroups are conjugate). 
Then $S$ admits a $G$-invariant equivalence relation 
$R\subset S\times S$ iff $\Stab(p)$ is \emph{not} maximal, 
that is, iff $\Stab(p)$ is properly contained in a proper 
subgroup $H$ of $G$: $\Stab(p)\subsetneq H\subsetneq G$. 
\end{proposition}
A proof of this may be found in \cite{Jacobson:BasicAlgebraI}
(Theorem\,1.12). Regarding the action of the inhomogeneous Galilei 
and Lorentz groups on spacetime, their stabilisers are the 
corresponding homogeneous groups. Now, the homogeneous Lorentz 
group is maximal in the inhomogeneous one, whereas the homogeneous 
Galilei group is not maximal in the inhomogeneous one, since 
it can still be supplemented by time translations without the 
need to also invoke space 
translations.\footnote{The homogeneous Galilei group only acts 
on the spatial translations, not the time translations, whereas the 
homogeneous Lorentz group acts irreducibly on the vector space 
of translations.} This, according to Proposition\,\ref{prop:AbsSimMath}, 
is the group theoretic origin of the absence of any invariant 
simultaneity structure in the Lorentzian case. 

However, one may ask whether there are simultaneity structures 
\emph{relative} to some \emph{additional} structure $X$. As additional 
structure, $X$, one could, for example, take an inertial reference 
frame, which is characterised by a foliation of spacetime by 
parallel timelike lines. The stabiliser subgroup of that 
structure within the proper orthochronous Poincar\'e group is 
given by the semidirect product of spacetime translations with 
all rotations in the hypersurfaces perpendicular to the lines in 
$X$:
\begin{equation}
\label{eq:Aut_X}
\mathrm{Stab}_X(\ILor_{\uparrow +})
\cong\reals^4\rtimes\group{SO}(3)\,.
\end{equation}
Here the $\group{SO}(3)$ only acts on the spatial translations,
so that the group is also isomorphic to $\reals\times\group{E}(3)$,
where $\group{E}(3)$ is the group of Euclidean motions in 
3-dimensions (the hyperplanes perpendicular to the lines in $X$). 
We can now ask: how many admissible 
$\mathrm{Stab}_X(\ILor_{\uparrow +})$ -- invariant equivalence 
relations are there. The answer is      
\begin{proposition}
\label{prop:RealSim}
There exits precisely one admissible simultaneity 
structure which is invariant under $\mathrm{Stab}_X(\ILor_{\uparrow +})$,
where $X$ represents am inertial reference frame (a foliation of
spacetime by parallel timelike lines). It is given by Einstein 
simultaneity, that is, the equivalence classes are the hyperplanes
perpendicular to the lines in $X$.
\end{proposition}
The proof is given in~\cite{Giulini:2002a}. Note again the 
connection to quoted group-theoretic result: The stabiliser 
subgroup of a point in $\mathrm{Stab}_X(\ILor_{\uparrow +})$ 
is $\group{SO(3)}$, which is clearly not maximal in 
$\mathrm{Stab}_X(\ILor_{\uparrow +})$ since it is a proper 
subgroup of $\group{E}(3)$ which, in turn, is a proper subgroup 
of $\mathrm{Stab}_X(\ILor_{\uparrow +})$.

\subsection{The lattices of causally and chronologically complete sets}
\label{sec:CausCompSets}
Here we wish to briefly discuss another important structure associated 
with causality relations in Minkowski space, which plays a fundamental 
r\^ole in modern Quantum Field Theory (see e.g. \cite{Haag:LocQuantPhys}). 
Let $S_1$ and $S_2$ be subsets of $\mathbb{M}^n$. We say that $S_1$ and 
$S_2$ are \emph{causally disjoint} or \emph{spacelike separated} iff $p_1-p_2$ 
is spacelike, i.e. $(p_1-p_2)^2<0$, for any $p_1\in S_1$ and $p_2\in S_2$. 
Note that because a point is not spacelike separated from itself, 
causally disjoint sets are necessarily disjoint in the 
ordinary set-theoretic sense---the converse being of course not true. 

For any subset $S\subseteq\mathbb{M}^n$ we denote by $S'$ the largest 
subset of $\mathbb{M}^n$ which is causally disjoint to $S$. The set 
$S'$ is called the \emph{causal complement} of $S$. The procedure of 
taking the causal complement can be iterated and we set $S'':=(S')'$
etc. 
$S''$ is called the \emph{causal completion} of $S$. It also follows 
straight from the definition that $S_1\subseteq S_2$ implies 
$S'_1\supseteq S'_2$ and 
also $S''\supseteq S$. If $S''=S$ we call $S$ \emph{causally complete}. 
We note that the causal complement $S'$ of any given $S$ is 
automatically causally complete. Indeed, from $S''\supseteq S$ we 
obtain  $(S')''\subseteq S'$, but the first inclusion applied to 
$S'$ instead of $S$ leads to $(S')''\supseteq S'$, showing $(S')''=S'$. 
Note also that for any subset $S$ its causal completion, $S''$, is the 
smallest causally complete subset containing $S$, for if 
$S\subseteq K\subseteq S''$ with $K''=K$, we derive from the first 
inclusion by taking ${}''$ that $S''\subseteq K$, so that the second 
inclusion yields $K=S''$. Trivial examples of causally complete 
subsets of $\mathbb{M}^n$ are the empty set, single points, and the 
total set $\mathbb{M}^n$. Others are the open diamond-shaped 
regions (\ref{eq:OpenDiamonds}) as well as their closed counterparts:      
\begin{equation}
\label{eq:ClosedDiamonds}
\bar{U}(p,q):=(\mathcal{\bar C}^+_p\cap\mathcal{\bar C}^-_q)\cup
        (\mathcal{\bar C}^+_q\cap\mathcal{\bar C}^-_p)\,.
\end{equation}

We now focus attention to the set $\mathrm{Caus}(\mathbb{M}^n)$ 
of causally complete subsets of $\mathbb{M}^n$, including the 
empty set, $\emptyset$, and the total set, $\mathbb{M}^n$, which are 
mutually causally complementary. It is partially ordered by ordinary  
set-theoretic inclusion $(\subseteq)$ (cf. Sect.\,\ref{sec:GroupActions})
and carries the `dashing operation' $(')$ of taking the causal complement.
Moreover, on $\mathrm{Caus}(\mathbb{M}^n)$ we can define 
the operations of `meet' and `join', denoted by $\wedge$ and $\vee$ 
respectively, as follows: Let $S_i\in\mathrm{Caus}(\mathbb{M}^n)$ where 
$i=1,2$, then $S_1\wedge S_2$ is the largest causally complete subset in the 
intersection $S_1\cap S_2$ and $S_1\vee S_2$ is the smallest 
causally complete set containing the union $S_1\cup S_2$.
 
The operations of $\wedge$ and $\vee$ can be characterised in terms 
of the ordinary set-theoretic intersection $\cap$ together with the 
dashing-operation. To see this, 
consider two causally complete sets, $S_i$ where $i=1,2$, and note 
that the set of points that are spacelike separated from $S_1$ and 
$S_2$ are obviously given by $S'_1\cap S'_2$, but also by 
$(S_1\cup S_2)'$, 
so that 
\begin{subequations}
\label{eg:CapCup}
\begin{alignat}{2}
\label{eq:CapCup1}
& S'_1\cap S'_2&&\,=\,(S_1\cup S_2)'\,,\\
\label{eq:CapCup2}
& S_1\cap S_2&&\,=\,(S'_1\cup S'_2)'\,.
\end{alignat}
\end{subequations}
Here (\ref{eq:CapCup1}) and (\ref{eq:CapCup2}) are equivalent 
since any $S_i\in\mathrm{Caus}(\mathbb{M}^n)$ can be written as 
$S_i=P'_i$, namely $P_i=S'_i$. If $S_i$ runs through all 
sets in $\mathrm{Caus}(\mathbb{M}^n)$ so does $P_i$. Hence any 
equation that holds generally for all $S_i\in\mathrm{Caus}(\mathbb{M}^n)$
remains valid if the $S_i$ are replaced by $S'_i$. 

Equation (\ref{eq:CapCup2}) immediately shows that $S_1\cap S_2$
is causally complete (since it is the $'$ of something). Taking the 
causal complement of (\ref{eq:CapCup1}) we obtain the desired relation 
for $S_1\vee S_2:=(S_1\cup S_2)''$. Together we have  
\begin{subequations}
\label{eg:MeetJoin}
\begin{alignat}{2}
\label{eg:MeetJoin1}
& S_1\wedge S_2&&\,=\,S_1\cap S_2\,,\\
\label{eg:MeetJoin2}
& S_1\vee S_2&&\,=\,(S'_1\cap S'_2)'\,.
\end{alignat}
\end{subequations}
From these we immediately derive  
\begin{subequations}
\label{eg:Negating}
\begin{alignat}{2}
\label{eg:Negating1}
& (S_1\wedge S_2)'&&\,=\,S'_1\vee S'_2\,,\\
\label{eg:Negating2}
& (S_1\vee S_2)'&&\,=\,S'_1\wedge S'_2\,.
\end{alignat}
\end{subequations}

All what we have said so far for the set $\mathrm{Caus}(\mathbb{M}^n)$
could be repeated verbatim for the set $\mathrm{Chron}(\mathbb{M}^n)$
of \emph{chronologically complete} subsets. We say that $S_1$ and 
$S_2$ are \emph{chronologically disjoint} or \emph{non-timelike separated}, 
iff $S_1\cap S_2=\emptyset$ and $(p_1-p_2)^2\leq 0$ for any $p_1\in S_1$ 
and $p_2\in S_2$. $S'$, the \emph{chronological complement} of $S$, is 
now the largest subset of $\mathbb{M}^n$ which is chronologically 
disjoint to $S$. The only difference between the causal and the 
chronological complement of $S$ is that the latter now contains 
lightlike separated points outside $S$. 
A set $S$ is \emph{chronologically complete} iff $S=S''$, where the 
dashing now denotes the operation of taking the chronological complement. 
Again, for any set $S$ the set $S'$ is automatically chronologically 
complete and $S''$ is the smallest  chronologically complete subset 
containing $S$. Single points are chronologically complete subsets.
All the formal properties regarding ${}'$, $\wedge$, and $\vee$ stated 
hitherto for $\mathrm{Caus}(\mathbb{M}^n)$ are the same for 
$\mathrm{Chron}(\mathbb{M}^n)$. 

One major difference between $\mathrm{Caus}(\mathbb{M}^n)$ and 
$\mathrm{Chron}(\mathbb{M}^n)$ is that the types of diamond-shaped 
sets they contain are different. For example, the closed ones, 
(\ref{eq:ClosedDiamonds}), are members of both. The open ones, 
(\ref{eq:OpenDiamonds}), are contained in $\mathrm{Caus}(\mathbb{M}^n)$ 
but \emph{not} in $\mathrm{Chron}(\mathbb{M}^n)$. Instead, 
$\mathrm{Chron}(\mathbb{M}^n)$ contains the closed diamonds whose 
`equator'\footnote{By `equator' we mean the $(n-2)$--sphere in which the 
forward and backward light-cones in (\ref{eq:ClosedDiamonds}) intersect. 
In the two-dimensional drawings the `equator' is represented by just 
two points marking the right and left corners of the diamond-shaped set.}
have been removed. An essential structural difference between 
$\mathrm{Caus}(\mathbb{M}^n)$ and $\mathrm{Chron}(\mathbb{M}^n)$
will be stated below, after we have introduced the notion of a 
lattice to which we now turn. 

To put all these formal properties into the right frame we recall 
the definition of a lattice. Let $(L,\leq)$ be a partially ordered set 
and $a,b$ any two elements in $L$. Synonymously with $a\leq b$ we 
also write $b\geq a$ and say that $a$ is smaller than $b$, $b$ is 
bigger than $a$, or $b$ majorises $a$. We also write $a<b$ if 
$a\leq b$ and $a\ne b$. If, with respect to $\leq$, their greatest 
lower and least upper bound exist, they are denoted by 
$a\wedge b$---called the `meet of $a$ and $b$'---and $a\vee b$---called 
the `join of $a$ and $b$'---respectively. A partially ordered set for 
which the greatest lower and least upper bound exist for any pair 
$a,b$ of elements from $L$ is called a \emph{lattice}.

We now list some of the most relevant additional structural elements 
lattices can have: A lattice is called \emph{complete} if greatest 
lower and least upper bound exist for any subset $K\subseteq L$. 
If $K=L$ they are called $0$ (the smallest element in the lattice)
and $1$ (the biggest element in the lattice) respectively. An 
\emph{atom} in a lattice is an element $a$ which majorises only $0$, 
i.e. $0\leq a$ and if $0\leq b\leq a$ then $b=0$ or $b=a$. The lattice 
is called \emph{atomic} if each of its elements different from $0$ 
majorises an atom. An atomic lattice is called \emph{atomistic} if 
every element is the join of the atoms it majorises. An element 
$c$ is said to \emph{cover} $a$ if $a<c$ and if $a\leq b\leq c$ 
either $a=b$ or $b=c$. An atomic lattice is said to have the 
\emph{covering property} if, for every element $b$ and every atom $a$
for which $a\wedge b=0$, the join $a\vee b$ covers $b$.  

The subset $\{a,b,c\}\subseteq L$ is called a 
\emph{distributive triple} if 
\begin{subequations}
\label{eq:DistTriple}
\begin{alignat}{3}
\label{eq:DistTriple1}
& a\wedge (b\vee c) &&\,=\, (a\wedge b)\vee(a\wedge c)
\quad&&\text{and $(a,b,c)$ cyclically permuted}\,,\\
\label{eq:DistTriple2}
& a\vee (b\wedge c) &&\,=\, (a\vee b)\wedge(a\vee c)
\quad&&\text{and $(a,b,c)$ cyclically permuted}\,.   
\end{alignat}
\end{subequations}
\begin{definition}
A lattice is called \emph{distributive} or \emph{Boolean} if every triple 
$\{a,b,c\}$ is distributive. It is called \emph{modular} if every triple
$\{a,b,c\}$ with $a\leq b$ is distributive.
\end{definition} 
It is straightforward to check from (\ref{eq:DistTriple})
that modularity is equivalent to the following single condition: 
\begin{equation}
\label{eq:Modularity}
\text{modularity}\Leftrightarrow 
a\vee (b\wedge c)=b\wedge(a\vee c)\quad
\text{for all $a,b,c\in L$ s.t. $a\leq b$.}
\end{equation}

If in a lattice with smallest element $0$ and greatest element $1$ 
a map $L\rightarrow L$, $a\mapsto a'$, exist such that 
\begin{subequations}
\label{eq:Orthocomplement}
\begin{alignat}{1}
\label{eq:Orthocomplement1}
& a'':=(a')'=a\,,\\
\label{eq:Orthocomplement2}
& a\leq b\Rightarrow b'\leq a'\,,\\
\label{eq:Orthocomplement3}
& a\wedge a'=0\,,\quad a\vee a'=1\,,
\end{alignat}
\end{subequations}
the lattice is called \emph{orthocomplemented}. It follows that whenever 
the meet and join of a subset $\{a_i\mid i\in I\}$ ($I$ is some index set) 
exist one has De\,Morgan's laws\footnote{From these laws it also appears that 
the definition (\ref{eq:Orthocomplement3}) is redundant, as each of its 
two statements follows from the other, due to $0'=1$.}:   
\begin{subequations}
\label{eq:GenDeMorgan}
\begin{alignat}{2}
\label{eq:GenDeMorgan1}
&\bigl(\textstyle{\bigwedge_{i\in I}}\ a_i\bigr)'&&\,=\,
\textstyle{\bigvee_{i\in I}}\ a'_i\,,\\
\label{eq:GenDeMorgan2}
&\bigl(\textstyle{\bigvee_{i\in I}}\ a_i\bigr)'&&\,=\,
\textstyle{\bigwedge_{i\in I}}\ a'_i\,.
\end{alignat}
\end{subequations}

For orthocomplemented lattices there is a still weaker version 
of distributivity than modularity, which turns out to be physically 
relevant in various contexts: 
\begin{definition}
\label{def:Orthomodular}
An orthocomplemented lattice is called \emph{orthomodular} if every 
triple $\{a,b,c\}$ with $a\leq b$ and $c\leq b'$ is distributive. 
\end{definition}
From (\ref{eq:Modularity}) and using that $b\wedge c=0$ for $b\leq c'$ 
one sees that this is equivalent to the single condition (renaming 
$c$ to $c'$): 
\begin{subequations}
\label{eq:Orthomodularity}
\begin{alignat}{3}
\label{eq:Orthomodularity1}
\text{orthomod.}\
&\Leftrightarrow \quad
&&a\,=\,b\wedge(a\vee c')
\quad&&\text{for all $a,b,c\in L$ s.t. $a\leq b\leq c$\,,}\\
\label{eq:Orthomodularity2}
&\Leftrightarrow 
&&a\,=\,b\vee(a\wedge c')
\quad&&\text{for all $a,b,c\in L$ s.t. $a\geq b\geq c$\,,}
\end{alignat}
\end{subequations}
where the second line follows from the first by taking its 
orthocomplement and renaming $a',b',c$ to $a,b,c'$. It turns out 
that these conditions can still be simplified by making them 
independent of $c$. In fact, (\ref{eq:Orthomodularity}) are 
equivalent to 
\begin{subequations}
\label{eq:AltOrthomodularity}
\begin{alignat}{3}
\label{eq:AltOrthomodularity1}
\text{orthomod.}\
&\Leftrightarrow \quad
&&a=b\wedge(a\vee b')
\quad&&\text{for all $a,b\in L$ s.t. $a\leq b$\,,}\\
\label{eq:AltOrthomodularity2}
&\Leftrightarrow 
&&a=b\vee(a\wedge b')
\quad&&\text{for all $a,b\in L$ s.t. $a\geq b$\,.}
\end{alignat}
\end{subequations}
It is obvious that (\ref{eq:Orthomodularity}) implies 
(\ref{eq:AltOrthomodularity}) (set $c=b$). But the converse is 
also true. To see this, take e.g. (\ref{eq:AltOrthomodularity2})
and choose any $c\leq b$. Then $c'\geq b'$, $a\geq b$ (by hypothesis), 
and $a\geq a\wedge c'$ (trivially), so that $a\geq b\vee(a\wedge c')$. 
Hence $a\geq b\vee(a\wedge c')\geq b\vee(a\wedge b')=a$, which proves 
(\ref{eq:Orthomodularity2}).

Complete orthomodular atomic lattices are automatically atomistic. 
Indeed, let $b$ be the join of all atoms majorised by $a\ne 0$. Assume
$a\ne b$ so that necessarily $b<a$, then (\ref{eq:AltOrthomodularity2}) 
implies $a\wedge b'\ne 0$. Then there exists an atom $c$ majorised by 
$a\wedge b'$. This implies $c\leq a$ and $c\leq b'$, hence also 
$c\not\leq b$. But this is a contradiction, since $b$ is by definition 
the join of all atoms majorised by $a$.    
 
Finally we mention the notion of \emph{compatibility} or \emph{commutativity}, 
which is a symmetric, reflexive, but generally not transitive relation 
$R$ on an orthomodular lattice (cf. Sec.\,\ref{sec:GroupActions}). 
We write $a\natural b$ for $(a,b)\in R$ and define: 
\begin{subequations}
\label{eq:DefCompatibility}
\begin{alignat}{3}
\label{eq:DefCompatibility1}
a\natural b\quad
&\Leftrightarrow\quad
&&a&&\,=\,(a\wedge b)\vee (a\wedge b')\,,\\
\label{eq:DefCompatibility2}
&\Leftrightarrow\quad
&&b&&\,=\,(b\wedge a)\vee (b\wedge a')\,.
\end{alignat}
\end{subequations}
The equivalence of these two lines, which shows that the relation of 
being compatible is indeed symmetric, can be demonstrated using 
orthomodularity as follows: Suppose (\ref{eq:DefCompatibility1}) holds; 
then $b\wedge a'=b\wedge(b'\vee a')\wedge(b\vee a')=b\wedge(b'\vee a')$,
where we used  the orthocomplement of (\ref{eq:DefCompatibility1}) to 
replace $a'$ in the first expression and the trivial identity 
$b\wedge(b\vee a')=b$ in the second step. Now, applying 
(\ref{eq:AltOrthomodularity2}) to $b\geq a\wedge b$ we get 
$b=(b\wedge a)\vee[b\wedge(b'\vee a')]=(b\wedge a)\vee (b\wedge a')$, 
i.e. (\ref{eq:DefCompatibility2}). The converse,  
$(\ref{eq:DefCompatibility2})\Rightarrow(\ref{eq:DefCompatibility1})$, 
is of course entirely analogous. 
 
From (\ref{eq:DefCompatibility}) a few things are immediate: 
$a\natural b$ is equivalent to $a\natural b'$, $a\natural b$ is 
implied by $a\leq b$ or $a\leq b'$, and the elements $0$ and $1$ 
are compatible with all elements in the lattice. The \emph{centre} 
of a lattice is the set of elements which are compatible with all 
elements in the lattice. In fact, the centre is a Boolean sublattice. 
If the centre contains no other elements than $0$ and $1$ the 
lattice is said to be \emph{irreducible}. The other extreme is a 
Boolean lattice, which is identical to its own centre. Indeed, 
if $(a,b,b')$ is a distributive triple, one has 
$a=a\wedge 1=a\wedge (b\vee b')=(a\wedge b)\vee (a\wedge b')
\Rightarrow(\ref{eq:DefCompatibility1})$.

After these digression into elementary notions of lattice theory we 
come back to our examples of the sets $\mathrm{Caus}(\mathbb{M}^n)$
$\mathrm{Chron}(\mathbb{M}^n)$. Our statements above amount 
to saying that they are complete, atomic, and orthocomplemented 
lattices. The partial order relation $\leq$ is given by $\subseteq$ 
and the extreme elements $0$ and $1$ correspond to the empty set 
$\emptyset$ and the total set $\mathbb{M}^n$, the points of which 
are the atoms. Neither the covering property nor modularity 
is shared by any of the two lattices, as can be checked by way of 
elementary counterexamples.\footnote{An immediate counterexample for 
the covering property is this: Take two timelike separated points 
(i.e. atoms) $p$ and $q$. Then $\{p\}\wedge\{q\}=\emptyset$ whereas 
$\{p\}\vee\{q\}$ is given by the closed diamond (\ref{eq:ClosedDiamonds}). 
Note that this is true in $\mathrm{Caus}(\mathbb{M}^n)$ \emph{and} 
$\mathrm{Chron}(\mathbb{M}^n)$. But, clearly, $\{p\}\vee\{q\}$ does 
not cover either $\{p\}$ or $\{q\}$.} In particular, neither of them 
is Boolean. However, in \cite{Cegla.Jadczyk:1977} it was shown that 
$\mathrm{Chron}(\mathbb{M}^n)$ is orthomodular; see also \cite{Casini:2002} 
which deals with more general spacetimes. Note that by the argument 
given above this implies that $\mathrm{Chron}(\mathbb{M}^n)$ is 
atomistic. In contrast, $\mathrm{Caus}(\mathbb{M}^n)$ is definitely 
\emph{not} orthomodular, as is e.g. seen by the counterexample given 
in Fig.\,\ref{fig:CountEx}.\footnote{\label{foot:Haag}
Regarding this point, there are some conflicting statements in the 
literature. The first edition of \cite{Haag:LocQuantPhys} states 
orthomodularity of $\mathrm{Chron}(\mathbb{M}^n)$ in 
Proposition\,4.1.3, which is removed in the second edition without 
further comment. The proof offered in the first edition uses 
(\ref{eq:AltOrthomodularity1}) as definition of orthomodularity, 
writing $K_1$ for $a$ and $K_2$ for b. The crucial step is the claim 
that any spacetime event in the set $K_2\wedge(K_1\vee K_2')$ lies in 
$K_2$ and that any causal line through 
it must intersect either $K_1$ or $K_2'$.  The last statement is, 
however, not correct since the join of two sets (here $K_1$ and $K_2'$) 
is generally larger than the domain of dependence of their ordinary 
set-theoretic union; compare Fig.\,\ref{fig:CountEx}. :
(Generally, the domain of dependence of a subset $S$ of spacetime $M$ 
is the largest subset $D(S)\subseteq M$ such that any inextensible 
causal curve that intersects $D(S)$ also intersects $S$.)} 
It is also not difficult to prove that $\mathrm{Chron}(\mathbb{M}^n)$ 
is irreducible.\footnote{In general spacetimes $M$, the failure of 
irreducibility of $\mathrm{Chron}(M)$ is directly related to the 
existence of closed timelike curves; see~\cite{Casini:2002}.}
\noindent
\begin{figure}[htb]
\begin{minipage}[b]{0.38\linewidth}
\centering\epsfig{figure=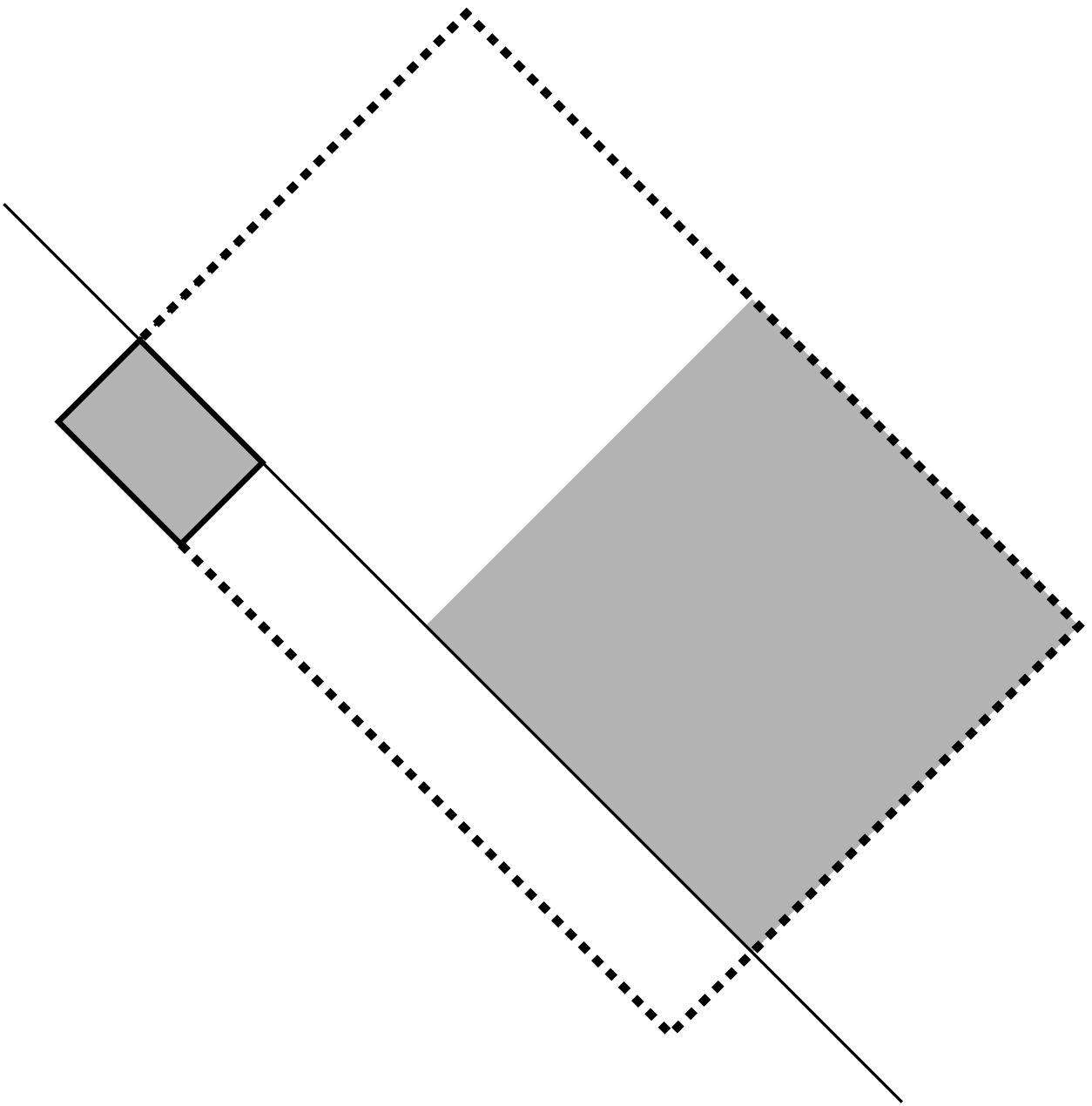,width=0.9\linewidth} 
\end{minipage}  
\hspace{0.8cm}
\begin{minipage}[b]{0.56\linewidth}
\centering\epsfig{figure=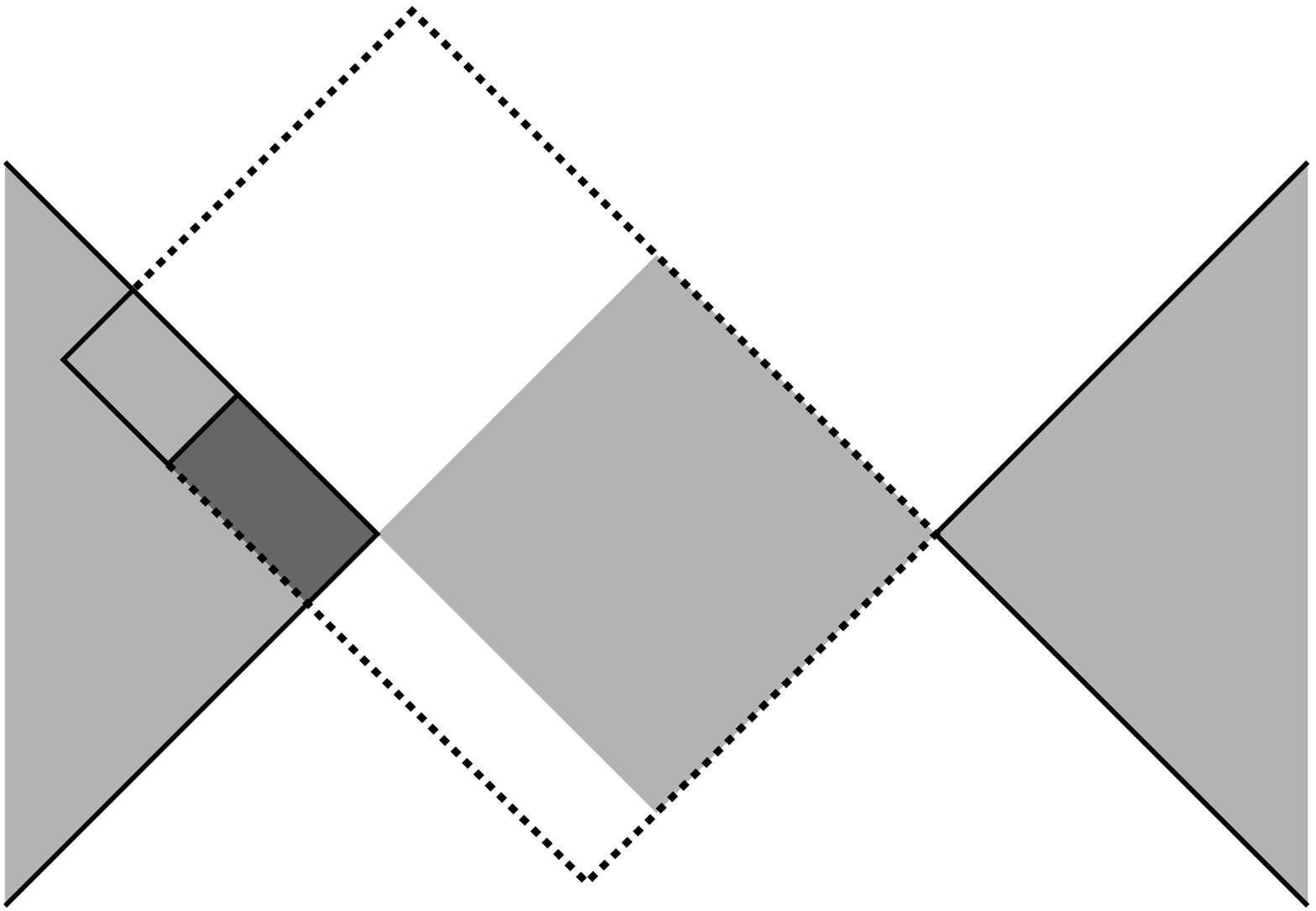,width=0.9\linewidth}
\end{minipage}
\put(-365,104){\large $\ell$}
\put(-344.5,72){\large $a$}
\put(-280,50){\large $b'$}
\put(-322,85){\large $a\vee b'$}
\put(-152,85){\large $a\vee b'$}
\put(-180,50){\large $b$}
\put(-30,50){\large $b$}
\put(-105,50){\large $b'$}
\put(-158,54){\large $a$}
\caption{\label{fig:CountEx}%
The two figures show that $\text{Caus}(\mathbb{M}^n)$ is not orthomodular. 
The first thing to note is that $\text{Caus}(\mathbb{M}^n)$ contains 
open (\ref{eq:OpenDiamonds}) as well as closed (\ref{eq:ClosedDiamonds})
diamond sets. In the left picture we consider the join of a small closed 
diamond $a$ with a large open diamond $b'$. (Closed sets are indicated by 
a solid boundary line.) 
Their edges are aligned along the lightlike line $\ell$. Even though 
these regions are causally disjoint, their causal completion is much 
larger than their union and given by the open (for $n>2$) enveloping 
diamond $a\vee b'$ framed by the dashed line.  (This also shows that 
the join of two regions can be larger than the domain of dependence 
of their union; compare footnote\,\ref{foot:Haag}.) . Next we consider 
the situation depicted on the right side. The closed double-wedge region 
$b$ contains the small closed diamond $a$. The causal 
complement $b'$ of $b$ is the open diamond in the middle. $a\vee b'$ is, 
according to the first picture, given by the large open diamond enclosed 
by the dashed line. The intersection of $a\vee b'$ with $b$ is strictly 
larger than $a$, the difference being the dark-shaded region in the left 
wedge of $b$ below $a$. Hence $a\ne b\wedge(a\vee b')$, in contradiction 
to (\ref{eq:AltOrthomodularity1}).}
\end{figure}

It is well known that the lattices of propositions for classical systems 
are Boolean, whereas those for quantum systems are merely orthomodular.
In classical physics the elements of the lattice are measurable subsets 
of phase space, with $\leq$ being ordinary set-theoretic inclusion 
$\subseteq$, and $\wedge$ and $\vee$ being ordinary set-theoretic 
intersection $\cap$ and union $\cup$ respectively. The orthocomplement 
is the ordinary set-theoretic complement.  In Quantum Mechanics 
the elements of the lattice are the closed subspaces of Hilbert space, 
with $\leq$ being again ordinary inclusion, $\wedge$ ordinary 
intersection, and $\vee$ is given by 
$a\vee b:=\overline{\Span\{a,b\}}$. The orthocomplement of a 
closed subset is the orthogonal complement in Hilbert space. For 
comprehensive discussions see \cite{Jauch:FoundQM} and 
\cite{Beltrametti:LogicQM}.

One of the main questions in the foundations of Quantum Mechanics is 
whether one could understand (derive) the usage of Hilbert spaces
and complex numbers from somehow more fundamental principles. Even 
though it is not a priori clear what ones measure of fundamentality 
should be at this point, an interesting line of attack consists in 
deriving the mentioned structures from the properties of the lattice 
of propositions (Quantum Logic). It can be shown that a lattice that 
is complete, atomic, irreducible, orthomodular, and that satisfies the 
covering property is isomorphic to the lattice of closed subspaces 
of a linear space with Hermitean inner product. The complex numbers 
are selected if additional technical assumptions are added. 
For the precise statements of these reconstruction theorems see
\cite{Beltrametti:LogicQM}.

It is now interesting to note that, on a formal level, there is a similar 
transition in going from Galilei invariant to Lorentz invariant causality 
relations. In fact, in Galilean spacetime one can also define a chronological 
complement: Two points are chronologically related if they are connected 
by a worldline of finite speed and, accordingly, two subsets in spacetime 
are chronologically disjoint if no point in one set is chronologically 
related to a point of the other. For example, the chronological 
complement of a point $p$ are all points simultaneous to, but different 
from, $p$. More general, it is not hard to see that the chronologically 
complete sets are just the subsets of some $t=\text{const.}$ hypersurface. 
The lattice of chronologically complete sets is then the continuous 
disjoint union of sublattices, each of which is isomorphic to the 
Boolean lattice of subsets in $\reals^3$. For details see 
\cite{Cegla.Jadczyk:1976}.    

As we have seen above, $\text{Chron}(\mathbb{M}^n)$ is complete, 
atomic, irreducible, and orthomodular (hence atomistic). The main 
difference to the lattice of propositions in Quantum Mechanics, as 
regards the formal aspects discussed here, is that 
$\text{Chron}(\mathbb{M}^n)$ does \emph{not} satisfy the covering 
property. Otherwise the formal similarities are intriguing and it 
is tempting to ask whether there is a deeper meaning to this. 
In this respect it would be interesting to know whether one could 
give a lattice-theoretic characterisation for $\text{Chron}(M)$ 
($M$ some fixed spacetime), comparable to the 
characterisation of the lattices of closed subspaces in Hilbert space 
alluded to above. Even for $M=\mathbb{M}^n$ such a characterisation 
seems, as far as I am aware, not to be known.

\subsection{Rigid motion}
\label{sec:RigidMotion}
As is well known, the notion of a rigid body, which proves so useful 
in Newtonian mechanics, is incompatible with the existence of a 
universal finite upper bound for all signal velocities~\cite{Laue:1911a}. 
As a result, the notion of a perfectly rigid body does not exist within 
the framework of SR. However, the notion of a \emph{rigid motion} does 
exist. Intuitively speaking, a body moves rigidly if, locally, the 
relative spatial distances of its material constituents are unchanging. 

The motion of an extended body is described by a normalised timelike 
vector field $u:\Omega\rightarrow\reals^n$, where $\Omega$ is an
open subset of Minkowski space, consisting of the events where the 
material body in question `exists'. We write $g(u,u)=u\cdot u=u^2$ 
for the Minkowskian scalar product. Being normalised now means that
$u^2=c^2$ (we do \emph{not} choose units such that $c=1$). The Lie derivative 
with respect to $u$ is denoted by $L_u$. 

For each material part of the body in motion its local rest space 
at the event $p\in\Omega$ can be identified with the hyperplane 
through $p$ orthogonal to $u_p$:
\begin{equation}
\label{eq:LocalRestSpace}
H_p:=p+u_p^\perp\,.
\end{equation}
$u_p^\perp$ carries a Euclidean inner product, $h_p$, given by the 
restriction of $-g$ to $u_p^\perp$. Generally we can write 
\begin{equation}
\label{eq:HorMetric}
h=c^{-2}\,u^\flat\otimes u^\flat-g\,,
\end{equation}
where $u^\flat=g^\downarrow(u):=g(u,\cdot)$ is the one-form associated 
to $u$. Following~\cite{Born:1909} the precise definition of `rigid  motion' 
can now be given as follows:  
\begin{definition}[Born 1909]
\label {def:RigidMotion}
Let $u$ be a normalised timelike vector field $u$. The motion described 
by its flow  is \emph{rigid} if
\begin{equation}
\label{eq:def:RigMo}
L_u h=0\,.
\end{equation}
\end{definition}
\noindent
Note that, in contrast to the Killing equations $L_ug=0$,
these equations are non linear due to the dependence of $h$ 
upon $u$.   

We write $\Proj_h:=\text{id}-c^{-2}\,u\otimes u^\flat\in\End(\reals^n)$ 
for the tensor field over spacetime that pointwise projects vectors 
perpendicular to $u$. It acts on one forms 
$\alpha$ via $\Proj_h(\alpha):=\alpha\circ\Proj_h$ and accordingly on all 
tensors. The so extended projection map will still be denoted by 
$\Proj_h$. Then we e.g. have 
\begin{equation}
\label{eq:ProjMetric}
h=-\Proj_h g:=-g(\Proj_h\cdot,\Proj_h\cdot)\,. 
\end{equation}

It is not difficult to derive the following two 
equations:\footnote{\label{foot:LieVel} Equation (\ref{eq:RigMo2})
simply follows from $L_u\Proj_h=-c^{-2}u\otimes L_u u^\flat$, so that 
$g((L_u\Proj_h)X,\Proj_h Y)=0$ for all $X,Y$. In fact, $L_uu^\flat=a^\flat$,
where $a:=\nabla_uu$ is the spacetime-acceleration. This follows from 
$L_uu^\flat(X)=L_u(g(u,X))-g(u,L_uX)=g(\nabla_uu,X)+g(u,\nabla_uX-[u,X])=
g(a,X)-g(u,\nabla_Xu)=g(a,X)$, where $g(u,u)=\text{const.}$ was used 
in the last step.}  
\begin{alignat}{2}
\label{eq:RigMo1}
& L_{fu}h&&\,=\,fL_uh\,,\\
\label{eq:RigMo2}
& L_uh&&\,=\,-L_u(\Proj_h g)=-\Proj_h(L_ug)\,,
\end{alignat}
where $f$ is any differentiable real-valued function on $\Omega$.

Equation (\ref{eq:RigMo1}) shows that the normalised vector field 
$u$ satisfies (\ref{eq:def:RigMo}) iff any rescaling $fu$ with a 
nowhere vanishing function $f$ does. Hence the normalization 
condition for $u$ in (\ref{eq:def:RigMo}) is really irrelevant.  
It is the geometry in spacetime of the flow lines and not their 
parameterisation which decide on whether motions (all, i.e. for any 
parameterisation, or none) along them are rigid. This has be the 
case because, generally speaking, there is no distinguished family 
of sections (hypersurfaces) across the bundle of flow lines that 
would represent `the body in space', i.e. mutually simultaneous 
locations of the body's points. Distinguished cases are those 
exceptional ones in which $u$ is hypersurface orthogonal. 
Then the intersection of $u$'s flow lines with the orthogonal 
hypersurfaces consist of mutually \emph{Einstein synchronous} 
locations of the points of the body. An example is discussed below. 

Equation (\ref{eq:RigMo2}) shows that the rigidity condition is 
equivalent to the `spatially' projected Killing equation. 
We call the flow of the timelike normalised vector field $u$ a 
\emph{Killing motion} (i.e. a spacetime isometry) if there is a 
Killing field $K$ such that $u=cK/\sqrt{K^2}$. Equation 
(\ref{eq:RigMo2}) immediately implies that Killing motions are 
rigid. What about the converse? Are there rigid motions that are 
not Killing? This turns out to be a difficult question. Its answer 
in Minkowski space is: `yes, many, but not as many as na\"{\i}vely 
expected.' 

Before we explain this, let us give an illustrative example for a 
Killing motion, namely that generated by the boost Killing-field
in Minkowski space. We suppress all but one spatial directions and 
consider boosts in $x$ direction in two-dimensional Minkowski space 
(coordinates $ct$ and $x$; metric $ds^2=c^2dt^2-dx^2$). The Killing 
field is\footnote{Here we adopt the standard notation from differential 
geometry, where $\partial_{\mu}:=\partial/\partial x^\mu$ denote the 
vector fields naturally defined by the coordinates 
$\{x^\mu\}_{\mu=0\cdots n-1}$. Pointwise the dual basis to 
$\{\partial_\mu\}_{\mu=0\cdots n-1}$ is $\{dx^\mu\}_{\mu=0\cdots n-1}$.}
\begin{equation}
\label{eq:BoostKill1}
K=x\,\partial_{ct}+ct\,\partial_x\,,
\end{equation}       
which is timelike in the region $\vert x\vert >\vert ct\vert$. 
We focus on the `right wedge' $x >\vert ct\vert$, which is 
now our region $\Omega$. Consider a rod of length $\ell$ which at $t=0$ 
is represented by the interval $x\in(r,r+\ell)$, where $r>0$. 
The flow of the normalised field $u=cK/\sqrt{K^2}$ is 
\begin{subequations}
\label{eq:BoostKill2}
\begin{alignat}{2}
\label{eq:BoostKill2a}
& ct(\tau)&&\,=\,x_0\,\sinh\bigl(c\tau/x_0)\,,\\
\label{eq:BoostKill2b}
&  x(\tau)&&\,=\,x_0\,\cosh\bigl(c\tau/x_0)\,,
\end{alignat}
\end{subequations}
where $x_0=x(\tau=0)\in(r,r+\ell)$ labels the elements of the 
rod at $\tau=0$. We have $x^2-c^2t^2=x_0^2$, showing that the 
individual elements of the rod move on hyperbolae (`hyperbolic 
motion'). $\tau$ is the proper time along each orbit, 
normalised so that the rod lies on the $x$ axis at $\tau=0$.

The combination 
\begin{equation}
\label{eq:KillingTime}
\lambda:=c\tau/x_0
\end{equation}
is just the flow parameter for $K$ (\ref{eq:BoostKill1}), sometimes 
referred to as `Killing time' (though it is dimensionless).  
From (\ref{eq:BoostKill2}) we can solve for $\lambda$ and $\tau$ 
as functions of $ct$ and $x$:
\begin{subequations}
\label{eq:LambdaTau}
\begin{alignat}{3}
\label{eq:LambdaTau1}
& \lambda&&\,=\,
f(ct,x)&&\,:=\,\tanh^{-1}\bigl(ct/x\bigr)\,,\\
\label{eq:LambdaTau2}
& \tau&&\,=\,
\hat f(ct,x)&&\,:=\,\underbrace{\sqrt{(x/c)^2-t^2}}_{x_0/c}\,\tanh^{-1}\bigl(ct/x\bigr)\,,
\end{alignat}
\end{subequations}
from which we infer that the hypersurfaces of constant $\lambda$ 
are hyperplanes which all intersect at the origin. Moreover, 
we also have $df=K^\flat/K^2$ ($d$ is just the ordinary exterior 
differential) so that the hyperplanes of constant 
$\lambda$ intersect all orbits of $u$ (and $K$) orthogonally.
Hence the hyperplanes of constant $\lambda$ qualify as the 
equivalence classes of mutually Einstein-simultaneous events in 
the region $x>\vert ct\vert$ for a family of observers moving along 
the Killing orbits. This does not hold for the hypersurfaces of 
constant $\tau$, which are curved.  

The modulus of the spacetime-acceleration (which is the same as the 
modulus of the spatial acceleration measured in the local rest frame) 
of the material part of the rod labelled by $x_0$ is 
\begin{equation}
\label{eq:AccRod}
\Vert a\Vert_g=c^2/x_0\,.
\end{equation} 
As an aside we generally infer from this that, given a timelike 
curve of local acceleration (modulus) $\alpha$, infinitesimally 
nearby orthogonal hyperplanes intersect at a spatial 
distance $c^2/\alpha$. This remark will become relevant in the 
discussion of part\,2 of the Noether-Herglotz theorem given below.

In order to accelerate the rod to the uniform velocity $v$ without
deforming it, its material point labelled by $x_0$ has to accelerate 
for the eigentime (this follows from (\ref{eq:BoostKill2}))
\begin{equation}
\label{eq:EigTimeRod}
\tau=\frac{x_0}{c}\tanh^{-1}(v/c)\,,
\end{equation} 
which depends on $x_0$. In contrast, the Killing time is the same 
for all material points and just given by the final rapidity. 
In particular, judged from the local observers moving with the rod, 
a rigid acceleration requires accelerating the rod's trailing end 
harder but shorter than pulling its leading end. 

In terms of the coordinates $(\lambda,x_0)$, which are co-moving with 
the flow of $K$, and $(\tau,x_0)$, which are co-moving with the flow 
of $u$, we just have $K=\partial/\partial\lambda$ and 
$u=\partial/\partial\tau$ respectively. The spacetime metric $g$ 
and the projected metric $h$ in terms of these coordinates are:
\begin{subequations}
\label{eq:MetricsComovCoord}
\begin{alignat}{2}
\label{eq:MetricsComovCoord1}
& h&&\,=\,dx_0^2\,,\\
\label{eq:MetricsComovCoord2}
& g&&\,=\,x_0^2\,d\lambda^2-dx_0^2
=c^2\bigl(d\tau-(\tau/x_0)\,dx_0\bigr)^2-dx_0^2\,.
\end{alignat}
\end{subequations}
Note the simple form $g$ takes in terms of $x_0$ and $\lambda$,
which are also called the `Rindler coordinates' for the region 
$\vert x\vert>\vert ct\vert$ of Minkowski space. They are the 
analogs in Lorentzian geometry to polar coordinates (radius $x_0$, 
angle $\lambda$) in Euclidean geometry.  

Let us now return to the general case. We decompose the derivative 
of the velocity one-form $u^\flat:=g^\downarrow(u)$ as follows:
\begin{equation}
\label{eq:DecVelOneform}
\nabla u^\flat=\theta+\omega+ c^{-2}\,u^\flat\otimes a^\flat\,,
\end{equation}
where $\theta$ and $\omega$ are the projected symmetrised and
antisymmetrised derivatives respectively\footnote{We denote the 
symmetrised and antisymmetrised tensor-product (not including 
the factor $1/n!$) by $\vee$ and $\wedge$ respectively and the 
symmetrised and antisymmetrised (covariant-) derivative by 
$\nabla\vee$ and $\nabla\wedge$. For example, 
$(u^\flat\wedge v^\flat)_{ab}=u_av_b-u_bv_a$ and 
$(\nabla\vee u^\flat)_{ab}=\nabla_a u_b+\nabla_bu_a$.
Note that $(\nabla\wedge u^\flat)$ is the same as the ordinary 
exterior differential $du^\flat$. Everything we say in the sequel 
applies to curved spacetimes if $\nabla$ is read as covariant 
derivative with respect to the Levi-Civita connection.}  
\begin{subequations}
\label{eq:ThetaOmega}
\begin{alignat}{3}
\label{eq:ThetaOmega1}
& 2\theta 
&&\,=\, \Proj_h(\nabla\vee u^\flat)
&&\,=\,\nabla\vee u^\flat-c^{-2}\,u^\flat\vee a^\flat\,,\\
\label{eq:ThetaOmega2}
& 2\omega 
&&\,=\,\Proj_h(\nabla\wedge u^\flat)
&&\,=\,\nabla\wedge u^\flat-c^{-2}\, u^\flat\wedge a^\flat\,.
\end{alignat}
\end{subequations}
The symmetric part, $\theta$, is usually further decomposed into 
its traceless and pure trace part, called the \emph{shear} and 
\emph{expansion} of $u$ respectively. The antisymmetric part 
$\omega$ is called the \emph{vorticity} of $u$. 

Now recall that the Lie derivative of $g$ is just twice the 
symmetrised derivative, which in our notation reads: 
\begin{equation}
\label{eq:LieDerNabla}
L_ug=\nabla\vee u^\flat\,.
\end{equation}
This implies in view of (\ref{eq:def:RigMo}), (\ref{eq:RigMo2}), and 
(\ref{eq:ThetaOmega1}) 
\begin{proposition}
Let $u$ be a normalised timelike vector field $u$. The motion described 
by its flow  is rigid iff $u$ is of vanishing shear and expansion, i.e. 
iff $\theta=0$. 
\end{proposition}  

Vector fields generating rigid motions are now classified according to 
whether or not they have a vanishing vorticity $\omega$: if $\omega=0$ 
the flow is called \emph{irrotational}, otherwise \emph{rotational}. 
The following theorem is due to Herglotz \cite{Herglotz:1910} 
and Noether~\cite{Noether:1910}:
\begin{theorem}[Noether \& Herglotz, part\,1]
\label{thm:NoetherHerglotz1}
A rotational rigid motion in Minkowski space must be a Killing motion.
\end{theorem}
An example of such a rotational motion is given by the Killing 
field\footnote{\label{foot:MetricCyl} We now use standard cylindrical 
coordinates $(z,\rho,\varphi)$, in terms of which  
$ds^2=c^2dt^2-dz^2-d\rho^2-\rho^2\,d\varphi^2$.}  
\begin{equation}
\label{eq:RotKill1}
K=\partial_t+\kappa\,\partial_{\varphi}
\end{equation}
inside the region 
\begin{equation}
\label{eq:RegionOmega}
\Omega=\{(t,z,\rho,\varphi)\mid\kappa\rho<c\}\,,
\end{equation}
where $K$ is timelike. This motion corresponds to a rigid rotation 
with constant angular velocity $\kappa$ which, without loss of generality, 
we take to be positive. Using the co-moving angular coordinate 
$\psi:=\varphi-\kappa t$, the split (\ref{eq:HorMetric}) is now 
furnished by 
\begin{subequations}
\label{eq:RotSplitMetric}
\begin{alignat}{2}
\label{eq:RotSplitMetric1}
& u^\flat&&\,=\,
c\,\sqrt{1-(\kappa\rho/c)^2}
\left\{
c\,dt-\frac{\kappa\rho/c}{1-(\kappa\rho/c)^2}\ \rho\,d\psi
\right\}\,,\\
\label{eq:RotSplitMetric2}
& h&&\,=\,
dz^2+d\rho^2+\frac{\rho^2\,d\psi^2}{1-(\kappa\rho/c)^2}\,.
\end{alignat}
\end{subequations}
The metric $h$ is curved (cf. Lemma\,\ref{lemma:NoetherHerglotz1-1}). 
But the rigidity condition (\ref{eq:def:RigMo}) means that $h$, and 
hence its curvature, cannot change along the motion. Therefore, even 
though we can keep a body in uniform rigid rotational motion, we cannot 
put it into this state from rest by purely rigid motions, since this would 
imply a transition from a flat to a curved geometry of the body. 
This was first pointed out by Ehrenfest~\cite{Ehrenfest:1909}. 
Below we will give a concise analytical expression of this fact 
(cf. equation~(\ref{eq:OmegaLieZero})). All this is in contrast to the 
translational motion, as we will also see below.

The proof of Theorem\,\ref{thm:NoetherHerglotz1} relies on 
arguments from differential geometry proper and is somewhat 
tricky. Here we present the essential steps, basically following 
\cite{PiraniWilliams:1962} and \cite{Trautman:1965} in a 
slightly modernised notation. Some straightforward calculational 
details will be skipped. The argument itself is best broken 
down into several lemmas. 

At the heart of the proof lies the following general 
construction: Let $M$ be the spacetime manifold 
with metric $g$ and $\Omega\subset M$ the open region in which the 
normalised vector field $u$ is defined. We take $\Omega$ to be simply 
connected. The orbits of $u$ foliate $\Omega$ and hence define an 
equivalence relation on $\Omega$ given by $p\sim q$ iff $p$ and $q$ lie 
on the same orbit. The quotient space $\hat\Omega:=\Omega/\!\!\sim$ is 
itself a manifold. Tensor fields on $\hat\Omega$ can be represented by 
(i.e. are in bijective correspondence to) tensor fields $T$ on 
$\Omega$ which obey the two conditions: 
\begin{subequations}
\label{eq:Projectibility}
\begin{alignat}{2}
\label{eq:Projectibility1}
&\Proj_h\, T&&\,=\,T\,,\\
\label{eq:Projectibility2}     
& L_uT&&\,=\,0\,.
\end{alignat}
\end{subequations}
Tensor fields satisfying (\ref{eq:Projectibility1}) are called 
\emph{horizontal}, those satisfying both conditions 
(\ref{eq:Projectibility}) are called \emph{projectable}. 
The $(n-1)$-dimensional metric tensor $h$, defined in 
(\ref{eq:HorMetric}), is an example of a projectable tensor if 
$u$ generates a rigid motion, as assumed here. It turns $(\hat\Omega,h)$ 
into a $(n-1)$-dimensional Riemannian manifold. The covariant derivative 
$\hat\nabla$ with respect to the Levi-Civita connection of $h$ is 
given by the following operation on projectable tensor fields: 
\begin{equation}
\label{eq:ProjCovDer}
\hat\nabla:=\Proj_h\circ\nabla
\end{equation}
i.e. by first taking the covariant derivative $\nabla$ (Levi-Civita
connection in $(M,g)$) in spacetime and then projecting the result 
horizontally. This results again in a projectable tensor, as a 
straightforward calculation shows. 

The horizontal projection of the spacetime curvature tensor can 
now be related to the curvature tensor of $\hat\Omega$ (which is a 
projectable tensor field). Without proof we state 
\begin{lemma}
\label{lemma:NoetherHerglotz1-1}
Let $u$ generate a rigid motion in spacetime. Then the horizontal 
projection of the totally covariant (i.e. all indices down) curvature 
tensor $R$ of $(\Omega,g)$ is related to the totally covariant 
curvature tensor $\hat R$ of $(\hat\Omega,h)$ by the following 
equation\footnote{$\hat R$ 
appears with a minus sign on the right hand side of 
(\ref{eq:ProjectedCurvatureRelation}) because the first index on 
the hatted curvature tensor is lowered with $h$ rather than $g$. 
This induces a minus sign due to (\ref{eq:HorMetric}), i.e. as a 
result of our `mostly-minus'-convention for the signature of the 
spacetime metric.}:
\begin{equation}
\label{eq:ProjectedCurvatureRelation} 
\Proj_h R=-\hat R-3\,(\text{\rm id}-\Proj_\wedge)\omega\otimes\omega\,,
\end{equation}
where $\Proj_\wedge$ is the total antisymmetriser, which here 
projects tensors of rank four onto their totally antisymmetric 
part. 
\end{lemma}
\noindent
Formula (\ref{eq:ProjectedCurvatureRelation}) is true in any spacetime 
dimension $n$. Note that the projector $(\text{id}-\Proj_\wedge)$
guarantees consistency with the first Bianchi identities for $R$ and 
$\hat R$, which state that the total antisymmetrisation in their last 
three slots vanish identically. This is consistent with 
(\ref{eq:ProjectedCurvatureRelation}) since for tensors of rank 
four with the symmetries of $\omega\otimes\omega$ the total 
antisymmetrisation on tree slots is identical to $\Proj_\wedge$,
the symmetrisation on all four slots. The claim now simply follows from 
$\Proj_\wedge\circ(\text{id}-\Proj_\wedge)=\Proj_\wedge-\Proj_\wedge=0$.

We now restrict to spacetime dimensions of four or less, i.e. $n\leq 4$.
In this case $\Proj_\wedge\circ\Proj_h=0$ since $\Proj_h$ makes the 
tensor effectively live over $n-1$ dimensions, and any totally 
antisymmetric four-tensor in three or less dimensions must vanish. 
Applied to (\ref{eq:ProjectedCurvatureRelation}) this 
means that $\Proj_\wedge(\omega\otimes\omega)=0$, for horizontality 
of $\omega$ implies $\omega\otimes\omega=\Proj_h(\omega\otimes\omega)$.
Hence the right hand side of  (\ref{eq:ProjectedCurvatureRelation}) 
just contains the pure tensor product $-3\,\omega\otimes\omega$. 

Now, in our case $R=0$ since $(M,g)$ is flat Minkowski space. 
This has two interesting consequences: First, $(\hat\Omega,h)$
is curved iff the motion is rotational, as exemplified above. 
Second, since $\hat R$ is projectable, its Lie derivative with 
respect to $u$ vanishes. Hence (\ref{eq:ProjectedCurvatureRelation}) 
implies $L_u\omega\otimes\omega+\omega\otimes L_u\omega=0$, which is 
equivalent to\footnote{In more than four spacetime dimensions one only 
gets $(\text{id}-\Proj_\wedge)(L_u\omega\otimes\omega+\omega\otimes L_u\omega)=0$.}
\begin{equation}
\label{eq:OmegaLieZero}  
L_u\omega=0\,.
\end{equation}
This says that the vorticity cannot change along a rigid motion 
in flat space. It is the precise expression for the remark above 
that you cannot rigidly \emph{set} a disk into rotation. Note 
that it also provides the justification for the global 
classification of rigid motions into rotational and irrotational 
ones. 

A sharp and useful criterion for whether a rigid motion is Killing
or not is given by the following 
\begin{lemma}
\label{lemma:NoetherHerglotz1-2}
Let $u$ be a normalised timelike vector field on a region 
$\Omega\subseteq M$. The motion generated by $u$ is Killing iff 
it is rigid and $a^\flat$ is exact on $\Omega$.  
\end{lemma}
\begin{proof}
That the motion generated by $u$ be Killing is equivalent to the 
existence of a positive function $f:\Omega\rightarrow\reals$ such 
that $L_{fu}g=0$, i.e. $\nabla\vee(fu^\flat)=0$. 
In view of (\ref{eq:ThetaOmega1}) this is equivalent to
\begin{equation}
\label{eq:Matteo1}
2\theta+(d\ln f+c^{-2} a^\flat)\vee u^\flat=0\,,
\end{equation}
which, in turn, is equivalent to $\theta=0$ and $a^\flat=-c^2\,d\ln f$. 
This is true since $\theta$ is horizontal, $\Proj_h\theta=\theta$, whereas 
the first term in (\ref{eq:Matteo1}) vanishes upon applying $\Proj_h$.
The result now follows from reading this equivalence both ways: 
1)~The Killing condition for $K:=fu$ implies rigidity for $u$
and exactness of $a^\flat$. 2)~Rigidity of $u$ and $a^\flat=-d\Phi$ 
imply that $K:=fu$ is Killing, where $f:=\exp(\Phi/c^2)$.  
\end{proof}

We now return to the condition (\ref{eq:OmegaLieZero})
and express $L_u\omega$ in terms of $du^\flat$. For this 
we recall that $L_uu^\flat=a^\flat$ 
(cf. footnote\,\ref{foot:LieVel}) and that Lie derivatives 
on forms commute with exterior derivatives\footnote{This is most 
easily seen by recalling that on forms the Lie derivative can be 
written as $L_u=d\circ i_u+i_u\circ d$, where $i_u$ is the map of 
inserting $u$ in the first slot.}. Hence we have 
\begin{equation}
\label{eq:LieOmega1}
2\,L_u\omega
=L_u(\Proj_h du^\flat)
=\Proj_h da^\flat
=da^\flat-c^{-2}u^\flat\wedge L_ua^\flat\,.
\end{equation}
Here we used the fact that the additional terms that result 
from the Lie derivative of the projection tensor $\Proj_h$ vanish, 
as a short calculation shows, and also that on forms the 
projection tensor $\Proj_h$ can be written as 
$\Proj_h=\text{id}-c^{-2}u^\flat\wedge i_u$, where $i_u$ denotes 
the map of insertion of $u$ in the first slot. 

Now we prove 
\begin{lemma}
\label{lemma:NoetherHerglotz1-3} 
Let $u$ generate a rigid motion in flat space such that 
$\omega\ne 0$, then  
\begin{equation}
\label{eq:LieAccZero}
L_ua^\flat=0\,.
\end{equation}
\end{lemma}
\begin{proof}
Equation (\ref{eq:OmegaLieZero}) says that $\omega$ is 
projectable (it is horizontal by definition). Hence 
$\hat\nabla\omega$ is projectable, which implies 
\begin{equation}
\label{eq:HatNablaOmega1}
L_u\hat\nabla\omega=0\,.
\end{equation} 
Using (\ref{eq:DecVelOneform}) with $\theta=0$ one has
\begin{equation}
\label{eq:HatNablaOmega2}
\hat\nabla\omega
=\Proj_h\nabla\omega
=\Proj_h\nabla\nabla u^\flat-c^{-2}\Proj_h(\nabla u^\flat\otimes a^\flat)\,.
\end{equation}
Antisymmetrisation in the first two tensor slots makes the first 
term on the right vanish due to the flatness on $\nabla$. The 
antisymmetrised right hand side is hence equal to 
$-c^{-2}\omega\otimes a^\flat$. Taking the Lie derivative of 
both sides makes the left hand side vanish due to 
(\ref{eq:HatNablaOmega1}), so that 
\begin{equation}
\label{eq:HatNablaOmega3}
L_u(\omega\otimes a^\flat)=\omega\otimes L_u a^\flat=0
\end{equation}
where we also used (\ref{eq:OmegaLieZero}). 
So we see that $L_ua^\flat=0$ if $\omega\ne 0$.\footnote{We 
will see below that (\ref{eq:LieAccZero}) is generally 
not true if $\omega=0$; see equation 
(\ref{eq:proof:NoeterHerglotz2i}).}
\end{proof}

The last three lemmas now constitute a proof for 
Theorem\,\ref{thm:NoetherHerglotz1}. Indeed, using 
(\ref{eq:LieAccZero}) in (\ref{eq:LieOmega1}) together with 
(\ref{eq:OmegaLieZero}) shows $da^\flat=0$, which,
according to Lemma\,\ref{lemma:NoetherHerglotz1-2}, implies 
that the motion is Killing. 

Next we turn to the second part of the theorem of Noether 
and Herglotz, which reads as follows:   
\begin{theorem}[Noether \& Herglotz, part\,2]
\label{thm:NoetherHerglotzP2}
All irrotational rigid motions in Minkowski space are given by 
the following construction: take a twice continuously differentiable 
curve $\tau\mapsto z(\tau)$ in Minkowski space, where w.l.o.g 
$\tau$ is the eigentime, so that ${\dot z}^2=c^2$. Let 
$H_\tau:=z(\tau)+({\dot z(\tau)})^\perp$ be the hyperplane through 
$z(\tau)$ intersecting the curve $z$  perpendicularly. Let $\Omega$ 
be a the tubular neighbourhood of $z$ in which no two hyperplanes 
$H_\tau,H_{\tau'}$ intersect for any pair $z(\tau),z(\tau')$ of 
points on the curve. In $\Omega$ define $u$ as the unique (once 
differentiable) normalised timelike vector field perpendicular to 
all $H_\tau\cap\Omega$. 
The flow of $u$ is the sought-for rigid motion.    
\end{theorem}
\begin{proof}
We first show that the flow so defined is indeed rigid, even though 
this is more or less obvious from its very definition, since we just 
defined it by `rigidly' moving a hyperplane through spacetime. 
In any case, analytically we have, 
\begin{equation}
\label{eq:proof:NoeterHerglotz2a}
H_\tau=\{x\in\mathbb{M}^n\mid f(\tau,x):=
{\dot z}(\tau)\cdot\bigl(x-z(\tau)\bigr)=0\}\,.
\end{equation}
In $\Omega$ any $x$ lies on exactly one such hyperplane, $H_\tau$,
which means that there is a function $\sigma:\Omega\rightarrow\reals$ 
so that $\tau=\sigma(x)$ and hence $F(x):=f(\sigma(x),x)\equiv 0$. 
This implies $dF=0$. Using the expression for $f$ from 
(\ref{eq:proof:NoeterHerglotz2a}) this is equivalent to 
\begin{equation}
\label{eq:proof:NoeterHerglotz2b}
d\sigma=\dot z^\flat\circ \sigma/[c^2-({\ddot z}\circ\sigma)
\cdot(\text{id}-z\circ\sigma)]\,,
\end{equation}
where `$\text{id}$' denotes the `identity vector-field', 
$x\mapsto x^\mu\partial_\mu$, in Minkowski space. Note that in 
$\Omega$ we certainly have $\partial_\tau f(\tau,x)\ne 0$ and hence 
$\ddot z\cdot(x-z)\ne c^2$. In $\Omega$ we now define the normalised 
timelike vector field\footnote{Note that, by definition of 
$\sigma$, $(\dot z\circ\sigma)\cdot(\text{id}-z\circ\sigma)\equiv 0$.} 
\begin{equation}
\label{eq:proof:NoeterHerglotz2c}
u:=\dot z\circ\sigma\,.
\end{equation}
Using (\ref{eq:proof:NoeterHerglotz2b}), its derivative is given by  
\begin{equation}
\label{eq:proof:NoeterHerglotz2d}
\nabla u^\flat=
d\sigma\otimes({\ddot z}^\flat\circ\sigma)=
\bigl[({\dot z}^\flat\circ\sigma)\otimes({\ddot z}^\flat\circ\sigma)
\bigr]/(N^2c^2)\,,
\end{equation}
where
\begin{equation}
\label{eq:proof:NoeterHerglotz2e}
N:=1-(\ddot z\circ\sigma)\cdot(\text{id}-z\circ\sigma)/c^2\,.
\end{equation}
This immediately shows that $\Proj_h\nabla u^\flat=0$ (since 
$\Proj_h{\dot z}^\flat=0$) and therefore that $\theta=\omega=0$. 
Hence $u$, as defined in (\ref{eq:proof:NoeterHerglotz2c}), 
generates an irrotational rigid motion.

For the converse we need to prove that any irrotational rigid 
motion is obtained by such a construction. So suppose $u$ is a
normalised timelike vector field such that $\theta=\omega=0$. 
Vanishing $\omega$ means $\Proj_h(\nabla\wedge u^\flat)=\Proj_h(d u^\flat)=0$. 
This is equivalent to $u^\flat\wedge du^\flat=0$, which according 
to the Frobenius theorem in differential geometry is equivalent 
to the integrability of the distribution\footnote{`Distribution'
is here used in the differential-geometric sense, where for a 
manifold $M$ it denotes an assignment of a linear subspace 
$V_p$ in the tangent space $T_pM$ to each point $p$ of $M$. 
The distribution $u^\flat=0$ is defined by  
$V_p=\{v\in T_pM\mid u^\flat_p(v)=u_p\cdot v=0\}$. A distribution 
is called (locally) integrable if (in the neighbourhood of each point)
there is a submanifold $M'$ of $M$ whose tangent space at any 
$p\in M'$ is just $V_p$.} 
$u^\flat=0$, i.e. the hypersurface orthogonality of $u$. We wish 
to show that the hypersurfaces orthogonal to $u$ are hyperplanes. 
To this end consider a spacelike curve $z(s)$, where $s$ is the 
proper length, running within one hypersurface perpendicular to $u$. 
The component of its second $s$-derivative parallel to the 
hypersurface is given by (to save notation we now simply write $u$ and 
$u^\flat$ instead of $u\circ z$ and $u^\flat\circ z$)  
\begin{equation}
\label{eq:proof:NoeterHerglotz2f}
\Proj_h\ddot z
=\ddot z-c^{-2}u\, u^\flat(\ddot z)  
=\ddot z+c^{-2}u\,\theta(\dot z,\dot z) =\ddot z\,,
\end{equation}
where we made a partial differentiation in the second step
and then used $\theta=0$. Geodesics in the hypersurface are 
curves whose second derivative with respect to proper length 
have vanishing components parallel to the hypersurface. Now, 
(\ref{eq:proof:NoeterHerglotz2f}) implies that geodesics in 
the hypersurface are geodesics in Minkowski space (the 
hypersurface is `totally geodesic'), i.e. given by straight 
lines. Hence the hypersurfaces are hyperplanes.
\end{proof}

Theorem\,\ref{thm:NoetherHerglotzP2} precisely corresponds to 
the Newtonian counterpart: The irrotational motion of a rigid 
body is determined by the worldline of any of its points, and 
any timelike worldline determines such a motion. We can rigidly 
put an extended body into any state of translational  motion, as 
long as the size of the body is limited by $c^2/\alpha$, where 
$\alpha$ is the modulus of its acceleration. This also shows 
that (\ref{eq:LieAccZero}) is generally not valid for 
irrotational rigid motions. In fact, the acceleration one-form
field for (\ref{eq:proof:NoeterHerglotz2c}) is 
\begin{equation}
\label{eq:proof:NoeterHerglotz2g}
a^\flat=(\ddot z^\flat\circ\sigma)/N
\end{equation}
from which one easily computes
\begin{equation}
\label{eq:proof:NoeterHerglotz2h}
da^\flat=(\dot z^\flat\circ\sigma)\wedge\left\{
(\Proj_h\dddot z^\flat\circ\sigma)+(\ddot z^\flat\circ\sigma)\frac{(\Proj_h\dddot z\circ\sigma)\cdot(\text{id}-z\circ\sigma)}{Nc^2}\right\}N^{-2}c^{-2}\,.
\end{equation}
From this one sees, for example, that for 
\emph{constant acceleration}, defined by $\Proj_h\dddot z=0$ (constant 
acceleration in time as measured in the instantaneous rest frame), 
we have $da^\flat=0$ and hence a Killing motion. Clearly, this is just 
the motion (\ref{eq:BoostKill2}) for the boost Killing field 
(\ref{eq:BoostKill1}). The Lie derivative of $a^\flat$ is now 
easily obtained: 
\begin{equation}
\label{eq:proof:NoeterHerglotz2i}
L_ua^\flat=i_uda^\flat
=(\Proj_h\dddot z^\flat\circ\sigma)N^{-2}\,,
\end{equation}
showing explicitly that it is not zero except for motions
of constant acceleration, which were just seen to be Killing 
motions. 

In contrast to the irrotational case just discussed, we have seen 
that we cannot put a body rigidly into rotational motion. 
In the old days this was sometimes expressed by saying that the rigid 
body in SR has only three instead of six degrees of freedom. 
This was clearly thought to be paradoxical as long as one assumed 
that the notion of a perfectly rigid body should also make sense in 
the framework of SR. However, this hope was soon realized to be 
physically untenable~\cite{Laue:1911a}.

\appendix
\section{Appendices}
In this appendix we spell out in detail some of the 
mathematical notions that were used in the main text.

\subsection{Sets and group actions}
\label{sec:GroupActions}
Given a set $S$, recall that an \emph{equivalence relation} is a 
subset $R\subset S\times S$ such that for all $p,q,r\in S$ the 
following conditions hold: 
1)~$(p,p)\in R$ (called `reflexivity'), 
2)~if $(p,q)\in R$ then $(q,p)\in R$ (called `symmetry'), and 
3)~if $(p,q)\in R$ and  $(q,r)\in R$ then $(p,r)\in R$ 
(called `transitivity'). Once $R$ is given, one often conveniently 
writes $p\sim q$ instead of $(p,q)\in R$. Given $p\in S$, its 
\emph{equivalence class}, $[p]\subseteq S$, is given by all 
points $R$-related to $p$, i.e. $[p]:=\{q\in S\mid (p,q)\in R\}$. 
One easily shows that equivalence classes are either identical or 
disjoint. Hence they form a \emph{partition} of $S$, that is, 
a covering by mutually disjoint subsets. Conversely, given a partition 
of a set $S$, it defines an equivalence relation by declaring two 
points as related iff they are members of the same cover set.  
Hence there is a bijective correspondence  between partitions of 
and equivalence relations on a set $S$. The set of equivalence 
classes is denoted by $S/R$ or $S/\!\!\sim$. There is a natural surjection 
$S\rightarrow S/R$, $p\mapsto [p]$.   

If in the definition of equivalence relation we exchange symmetry
for antisymmetry, i.e. $(p,q)\in R$ and $(q,p)\in R$ implies $p=q$, 
the relation is called a \emph{partial order}, usually written as 
$p\geq q$ for $(p,q)\in R$. If, instead, reflexivity is dropped 
and symmetry is replaced by asymmetry, i.e.  $(p,q)\in R$ 
implies $(q,p)\not\in R$, one obtains a relation called a 
\emph{strict partial order}, usually denoted by $p> q$ for 
$(p,q)\in R$.  
   
An \emph{left action} of a group $G$ on a set $S$ is a map
$\phi:G\times S\rightarrow S$, such that $\phi(e,s)=s$ 
($e=$ group identity) and $\phi(gh,s)=\phi(g,\phi(h,s))$. 
If instead of the latter equation we have $\phi(gh,s)=\phi(h,\phi(g,s))$
one speaks of a \emph{right action}. For left actions one sometimes 
conveniently writes $\phi(g,s)=:g\cdot s$, for right actions 
$\phi(g,s)=:s\cdot g$. An action is called \emph{transitive} if 
for every pair $(s,s')\in S\times S$ there is a $g\in G$ such that 
$\phi(g,s)=s'$, and \emph{simply transitive} if, in addition,  
$(s,s')$ determine $g$ uniquely, that is, $\phi(g,s)=\phi(g',s)$
for some $s$ implies $g=g'$. The action is called \emph{effective} 
if $\phi(g,s)=s$ for all $s$ implies $g=e$ (`every $g\ne e$ moves 
something') and \emph{free} if $\phi(g,s)=s$ for some $s$ implies 
$g=e$ (`no $g\ne e$ has a fixed point'). It is obvious that simple 
transitivity implies freeness and that, conversely, freeness and 
transitivity implies simple transitivity. Moreover, for Abelian groups, 
effectivity and transitivity suffice to imply simple transitivity. 
Indeed, suppose $g\cdot s=g'\cdot s$ holds for some $s\in S$, then 
we also have  $k\cdot(g\cdot s)=k\cdot(g'\cdot s)$ for all $k\in G$
and hence $g\cdot(k\cdot s)=g'\cdot(k\cdot s)$ by commutativity.
This implies that $g\cdot s=g'\cdot s$ holds, in fact, for all $s$. 

For any $s\in S$ we can consider the \emph{stabilizer subgroup}
\begin{equation}
\label{eq:def:Stab}
\Stab(s):=\{g\in G\mid\phi(g,s)=s\}\subseteq G\,.
\end{equation} 
If $\phi$ is transitive, any two stabilizer subgroups are conjugate: 
$\Stab(g\cdot s)=g\Stab(s)g^{-1}$. By definition, if $\phi$ is free 
all stabilizer subgroups are trivial (consist of the identity element 
only). In general, the intersection 
$G':=\bigcap_{s\in S}\Stab(s)\subseteq G$ is the normal subgroup of 
elements acting trivially on $S$. If $\phi$ is an action of $G$ on 
$S$, then there is an effective action $\hat\phi$ of $\hat G:= G/G'$ 
on $S$, defined by $\hat\phi([g],s):=\phi(g,s)$, where $[g]$ denotes 
the $G'$-coset of $G'$ in $G$.  

The \emph{orbit} of $s$ in $S$ under the action $\phi$ of $G$ 
is the subset
\begin{equation}
\label{eq:def:Orb}
\Orb(s):=\{\phi(g,s)\mid g\in G\}\subseteq S\,.
\end{equation} 
It is easy to see that group orbits are either disjoint or 
identical. Hence they define a partition of $S$, that is, 
an equivalence relation.   

A relation $R$ on $S$ is said to be invariant under the self map 
$f:S\rightarrow S$ if $(p,q)\in R\Leftrightarrow (f(p),f(q))\in R$.
It is said to be invariant under the action $\phi$ of $G$ on $S$ if 
$(p,q)\in R\Leftrightarrow (\phi(g,p),\phi(g,q))\in R$ for all 
$g\in G$. If $R$ is such a $G$-invariant equivalence relation, there is 
an action $\phi'$ of $G$ on the set $S/R$ of equivalence classes, 
defined by $\phi'(g,[p]):=[\phi(g,p)]$. A general theorem states that 
invariant equivalence relations exist for transitive group actions, 
iff the stabilizer subgroups (which in the transitive case are all 
conjugate) are maximal (e.g. Theorem\,1.12 in \cite{Jacobson:BasicAlgebraI}).

\subsection{Affine spaces}
\label{sec:AffineSpaces}
\begin{definition}
\label{def:AffineSpace}
An $n$-dimensional \textbf{affine space} over the field $\field$ 
(usually $\reals$ or $\complex$) is a triple $(S,V,\Phi)$,
where $S$ is a non-empty set, $V$ an $n$-dimensional vector
space over $\field$, and $\Phi$ an effective and transitive action 
$\Phi:V\times S\rightarrow S$ of $V$ (considered as Abelian group 
with respect to addition of vectors) on $S$. 
\end{definition}
We remark that an effective and transitive action of an Abelian 
group is necessarily simply transitive. Hence, without loss of 
generality, we could have required a simply transitive action in 
Definition\,\ref{def:AffineSpace} straightaway. We also note 
that even though the action $\Phi$ only refers to the Abelian group 
structure of $V$, it is nevertheless important for the definition 
of an affine space that $V$ is, in fact, a vector space (see below).
Any ordered pair of points $(p,q)\in S\times S$ uniquely defines 
a vector $v$, namely that for which $p=q+v$. It can be thought 
of as the difference vector pointing from $q$ to $p$. We write 
$v=\Delta(q,p)$, where $\Delta:S\times S\rightarrow V$ is a map 
which satisfies the conditions 
\begin{subequations}
\label{eq:AffDiffMap}
\begin{alignat}{2}
\label{eq:AffDiffMap1}
& \Delta(p,q)+\Delta(q,r)=\Delta(p,r)\qquad
&&\text{for all}\quad p,q,r\in S\,,\\
\label{eq:AffDiffMap2}
& \Delta_q:p\ni S\mapsto\Delta(p,q)\in V\quad
\text{is a bijection}\qquad
&&\text{for all}\quad p\in S\,.
\end{alignat}
\end{subequations}
Conversely, these conditions suffice to characterise an affine 
space, as stated in the following proposition, the proof of which 
is left to the reader: 
\begin{proposition}
\label{thm:AltDefAffSpace}
Let $S$ be a non-empty set, $V$ an $n$-dimensional vector space 
over $\field$ and $\Delta:S\times S\rightarrow V$ a map satisfying 
conditions~(\ref{eq:AffDiffMap}). Then $S$ is an $n$-dimensional 
affine space over $\field$ with action $\Phi(v,p):=\Delta^{-1}_p(v)$.
\end{proposition}

One usually writes $\Phi(v,p)=:p+v$, which defines what is meant 
by `$+$' between an element of an affine space and an element of $V$. 
Note that addition of two points in affine space is not defined. 
The property of being an action now states $p+0=p$ and $(p+v)+w=p+(v+w)$, 
so that in the latter case we may just write $p+v+w$. Similarly we 
write $\Delta(p,q)=:q-p$, defining what is meant by `$-$' between 
two elements of affine space. The minus sign also makes sense between 
an element of affine space and an element of vector space if one 
defines $p+(-v)=:p-v$. We may now write equations like 
\begin{equation}
\label{eq:AffDistributivity}
p+(q-r)=q+(p-r)\,,
\end{equation}
the formal proof of which is again left to the reader. 
It implies that 

Considered as Abelian group, any linear subspace $W\subset V$ 
defines a subgroup. The orbit of that subgroup in $S$ 
through $p\in S$ is an affine subspace, denoted by $W_p$, 
i.e.  
\begin{equation}
\label{eq:def:AffPlane}
W_p=p+W:=\{p+w\mid w\in W\}\,,
\end{equation}
which is an affine space over $W$ in its own right of dimension 
$\mathrm{dim}(W)$. One-dimensional affine subspaces are called 
\emph{(straight) lines}, two-dimensional ones \emph{planes}, and 
those of co-dimension one are called \emph{hyperplanes}.

\subsection{Affine maps}
\label{sec:AffineMaps}
Affine morphisms, or simply affine maps, are structure preserving maps 
between affine spaces. To define them in view of 
Definition\,\ref{def:AffineSpace} we recall once more the significance
of $V$ being a vector space and not just an Abelian group. This enters 
the following definition in an essential way, since there are 
considerably more automorphisms of $V$ as Abelian group, i.e. maps 
$f:V\rightarrow V$ that satisfy $f(v+w)=f(v)+f(w)$ for all $v,w\in V$,
than automorphisms of $V$ as linear space which, in addition, need to 
satisfy $f(av)=af(v)$ for all $v\in V$ and all 
$a\in\field$). In fact, the difference is precisely that the latter 
are all continuous automorphisms of $V$ (considered as topological 
Abelian group), whereas there are plenty (uncountably many) 
discontinuous ones, see \cite{Hamel:1905}.\footnote{
Let $\field=\reals$, then it is easy to see that $f(v+w)=f(v)+f(w)$ 
for all $v,w\in V$ implies $f(av)=af(v)$ for all $v\in V$ and all 
$a\in\rationals$ (rational numbers). For continuous $f$ this implies 
the same for all $a\in\reals$. All discontinuous $f$ are obtained 
as follows: let $\{e_\lambda\}_{\lambda\in I}$ be a (necessarily 
uncountable) basis of $\reals$ as vector space over $\rationals$ 
(`Hamel basis'), prescribe any values $f(e_\lambda)$, and extend 
$f$ linearly to all of $\reals$. Any value-prescription for which 
$I\ni\lambda\mapsto f(e_\lambda)/e_\lambda\in\reals$ is not 
constant gives rise to a non $\reals$-linear and discontinuous $f$. 
Such $f$ are `wildly' discontinuous in the following sense: for 
\emph{any} interval $U\subset\reals$, $f(U)\subset\reals$ is 
dense~\cite{Hamel:1905}.}

\begin{definition}
\label{def:AffineMap}
Let $(S,V,\Phi)$ and $(S',V',\Phi')$ be two affine spaces. 
An \textbf{affine morphism} or \textbf{affine map} is a pair 
of maps $F:S\rightarrow S'$ and $f:V\rightarrow V'$, where $f$ 
is linear, such that 
\begin{equation}
\label{eq:AffineMap1}
F\circ\Phi=\Phi'\circ f\times F\,.
\end{equation}
\end{definition}
In the convenient way of writing introduced above, this is 
equivalent to 
\begin{equation}
\label{eq:AffineMap2}
F(q+v)=F(q)+f(v)\,,
\end{equation}
for all $q\in S$ and all $v\in V$. (Note that the $+$ sign 
on the left refers to the action $\Phi$ of $V$ on $S$, 
whereas that on the right refers to the action $\Phi'$ of 
$V'$ on $S'$.) This shows that an affine map $F$ is determined 
once the linear map $f$ between the underlying 
vector spaces is given and the image $q'$ of an arbitrary point 
$q$ is specified. Equation (\ref{eq:AffineMap2}) can be rephrased 
as follows:  
\begin{corollary}
Let $(S,V,\Phi)$ and $(S',V',\Phi')$ be two affine spaces. 
A map $F:S\rightarrow S'$ is affine iff each of its 
restrictions to lines in $S$ is affine. 
\end{corollary}

Setting $p:=q+v$ equation (\ref{eq:AffineMap2}) is 
equivalent to 
\begin{equation}
\label{eq:AffineMap3}
F(p)-F(q)=f(p-q)
\end{equation}
for all $p,q\in S$. In view of the alternative definition of 
affine spaces suggested by Proposition\,\ref{thm:AltDefAffSpace}, 
this shows that we could have defined affine maps alternatively 
to (\ref{eq:AffineMap1}) by ($\Delta':S'\times S'\rightarrow V'$ 
is the difference map in $S'$)
\begin{equation}
\label{eq:AffineMap4}
\Delta'\circ F\times F=f\circ\Delta\,.
\end{equation}

Affine bijections of an affine space $(S,V,\Phi)$ onto itself 
form a group, the \emph{affine group}, denoted by 
$\group{GA}(S,V,\Phi)$. Group multiplication 
is just given by composition of maps, that is 
$(F_1,f_1)(F_2,f_2):=(F_1\circ F_2\,,\,f_1\circ f_2)$. 
It is immediate that the composed maps again 
satisfy~(\ref{eq:AffineMap1}). 

For any $v\in V$, the map $F=\Phi_v:p\mapsto p+v$ is an affine 
bijection for which $f=\Identity_V$. Note that in this case 
(\ref{eq:AffineMap1}) simply turns into the requirement  
$\Phi_v\circ\Phi_w=\Phi_w\circ\Phi_v$ for all $w\in V$, which 
is clearly satisfied due to $V$ being a commutative group. 
Hence there is a natural embedding $T:V\rightarrow\group{GL}(S,V,\Phi)$, 
the image $T(V)$of which is called the subgroup of 
\emph{translations}. The map $F\mapsto F_*:=f$ defines a 
group homomorphism $\group{GA}(S,V,\Phi)\rightarrow\group{GL}(V)$, since 
$(F_1\circ F_2)_*=f_1\circ f_2$. We have just seen that 
the translations are in the kernel of this map. In fact, 
the kernel is equal to the subgroup $T(V)$ of translations, 
as one easily infers from (\ref{eq:AffineMap3}) with 
$f=\Identity_V$, which is equivalent to $F(p)-p=F(q)-q$ 
for all $p,q\in S$. Hence there exists a $v\in V$ such 
that for all $p\in S$ we have $F(p)=p+v$. 

The quotient group $\group{GA}(S,V,\Phi)/T(V)$ is then clearly 
isomorphic to $\group{GL}(V)$. There are also embeddings 
$\group{GL}(V)\rightarrow\group{GA}(S,V,\Phi)$, but no canonical one:
each one depends on the choice of a reference point $o\in S$,
and is given by $\group{GL}(V)\ni f\mapsto F\in\group{GA}(S,V,\Phi)$, 
where $F(p):=o+f(p-o)$ for all $p\in S$. This shows that  
$\group{GA}(S,V,\Phi)$ is isomorphic to the semi-direct product 
$V\rtimes\group{GL}(V)$, though the isomorphism depends on the 
choice of $o\in S$. The action of $(a,A)\in V\rtimes\group{GL}(V)$ 
on $p\in S$ is then defined by 
\begin{equation}
\label{eq:AffineGroupAction}  
\bigl((a,A)\,,\,p\bigr)\mapsto o+a+A(p-o)\,,
\end{equation}    
which is easily checked to define indeed an 
($o$ dependent) action of $V\rtimes\group{GL}(V)$ on $S$.

\subsection{Affine frames, active and passive transformations}
\label{sec:AffineFrames}
Before giving the definition of an affine frame, we recall that of 
a linear frame: 
\begin{definition}
\label{def:LinearFrame}
A \textbf{linear frame} of the $n$-dimensional vector space $V$ over 
$\field$ is a basis $f=\{e_a\}_{a=1\cdots n}$ of $V$,
regarded as a linear isomorphism $f:\field^n\rightarrow V$,
given by $f(v^1,\cdots,v^n):= v^ae_a$. The set of linear frames 
of $V$ is denoted by $\Frames{V}$. 
\end{definition}
Since $\field$ and hence $\field^n$ carries a natural topology,
there is also a natural topology of $V$, namely that which makes 
each frame-map $f:\field^n\rightarrow V$ a homeomorphism. 

There is a natural right action of $\group{GL}(\field^n)$ on $\Frames{V}$, 
given by $(A,f)\rightarrow f\circ A$. It is immediate that this 
action is simply transitive. It is sometimes called the 
\emph{passive interpretation} of the transformation group 
$\group{GL}(\field^n)$, presumably because it moves the frames---associated 
to the observer---and not the points of $V$. 

On the other hand, any frame $f$ induces an isomorphism of 
algebras $\End(\field^n)\rightarrow\End(V)$, given by 
$A\mapsto A^f:=f\circ A\circ f^{-1}$. If $A=\{A^b_a\}$, then 
$A^f(e_a)=A^b_ae_b$, where $f=\{e_a\}_{a=1\cdots n}$. 
Restricted to $\group{GL}(\field^n)\subset\End(\field^n)$, this 
induces a group isomorphism $\group{GL}(\field^n)\rightarrow\group{GL}(V)$
and hence an $f$-dependent action of $\group{GL}(\field^n)$ on $V$
by linear transformations, defined by $(A,v)\mapsto A^fv=f(Ax)$, 
where $f(x)=v$. This is sometimes called the 
\emph{active interpretation} of the transformation group 
$\group{GL}(\field^n)$, presumably because it really moves the 
points of $V$. 

We now turn to affine spaces: 
\begin{definition}
\label{def:AffineFrame}
An \textbf{affine frame} of the $n$-dimensional affine space 
$(S,V,\Phi)$ over $\field$ is a tuple $F:=(o,f)$, where $o$ 
is a base point in $S$ and $f:\field^n\rightarrow V$ is a 
linear frame of $V$. $F$ is regarded as a map 
$\field^n\rightarrow S$, given by $F(x):=o+f(x)$. 
We denote the set of affine frames by $\Frames{(S,V,\Phi)}$.  
\end{definition}
Now there is a natural topology of $S$, namely that which makes 
each frame-map $F:\field^n\rightarrow S$ a homeomorphism. 

If we regard $\field^n$ as an affine space $\Aff(\field)$, 
it comes with a distinguished base point $o$, the zero 
vector. The group $\group{GA}\bigl(\Aff(\field^n)\bigr)$ is therefore 
naturally isomorphic to $\field^n\rtimes\group{GL}(\field^n)$. 
The latter naturally acts on $\field^n$ in the standard way,
$\Phi:((a,A),x)\mapsto \Phi((a,A),x):=A(x)+a$, where group 
multiplication is given by  
\begin{equation}
\label{eq:GeneralAffineGroup}
(a_1,A_1)(a_2,A_2)=(a_1+A_1a_2\,,\,A_1A_2)\,.
\end{equation}
The group $\field^n\rtimes\group{GL}(\field^n)$ has a natural right 
action on $\Frames{(S,V,\Phi)}$, where 
$(g,F)\mapsto F\cdot g:=F\circ g$. Explicitly, for $g=(a,A)$ 
and $F=(o,f)$, this action reads:  
\begin{equation}
\label{eq:AffActionPassive}
F\cdot g=(o,f)\cdot (a,A)=(o+f(a),f\circ A)\,.
\end{equation}
It is easy to verify directly that this is an action which, 
moreover, is again simply transitive. It is referred to as 
the \emph{passive interpretation} of the affine group 
$\field^n\rtimes\group{GL}(\field^n)$. 

Conversely, depending on the choice of an affine frame 
$F\in\Frames{(S,V,\Phi)}$, there is a group isomorphism 
$\field^n\rtimes\group{GL}(\field^n)\rightarrow\group{GA}(S,V,\Phi)$,
given by $(a,A)\mapsto F\circ (a,A)\circ F^{-1}$, and 
hence an $F$ dependent action of $\field^n\rtimes\group{GL}(\field^n)$
by affine maps on $(S,V,\Phi)$. If $F=(o,f)$ and $F(x)=p$, 
the action reads 
\begin{equation}
\label{eq:AffActionActive}
\bigl((a,A),p\bigr)\mapsto F(Ax+a)=A^f(p-o)+o+f(a)\,.
\end{equation}
This is called the \emph{active interpretation} of the 
affine group $\field^n\rtimes\group{GL}(\field^n)$.      

An affine frame $(o,f)$ with $f=\{e_a\}_{a=1\cdots n}$ defines 
$n+1$ points $\{p_0,p_1,\cdots p_n\}$, where $p_0:=o$
and $p_a:=o+e_a$ for $1\leq a\leq n$. Conversely, any 
$n+1$ points $\{p_0,p_1,\cdots p_n\}$ in affine space, for 
which $e_i:=p_i-p_0$ are linearly independent, define an affine 
frame. Note that this linear independence does not depend on 
the choice of $p_0$ as our base point, as one easily sees from 
the identity
\begin{equation}
\label{eq:AffFrameBasePointInd}
\sum_{a=1}^m v^a(p_a-p_0)=\sum_{k\ne a=0}^m v^a(p_a-p_k)\,,
\quad\mathrm{where}\quad
v^0:=-\sum_{a=1}^m v^a\,,
\end{equation}
which holds for any set $\{p_0,p_1,\cdots,p_m\}$ of $m+1$
points in affine space. To prove it one just needs
(\ref{eq:AffDistributivity}). Hence we say that these points are 
\emph{affinely independent} iff, e.g., the set of $m$ vectors 
$\{e_a:=p_a-p_0\mid 1\leq a\leq m\}$ is linearly independent. 
Therefore, an affine frame of $n$-dimensional affine space is 
equivalent to $n+1$ affinely independent points. Such a set of
points is also called an \emph{affine basis}.

Given an affine basis $\{p_0,p_1,\cdots,p_n\}\subset S$ 
and a point $q\in S$, there is a unique $n$-tuple 
$(v_1,\cdots ,v_n)\in\field^n$ such that 
\begin{subequations}
\label{eq:AffBasisExp}
\begin{equation}
\label{eq:AffBasisExp1}
q=p_0+\sum_{a=1}^n v^a(p_a-p_0)\,.
\end{equation}
Writing $v^k(p_k-p_0)=(p_k-p_0)+(1-v^k)(p_0-p_k)$ for some chosen 
$k\in\{1,\cdots,n\}$ and $v^a(p_a-p_0)=v^a(p_a-p_k)-v^a(p_0-p_k)$ 
for all $a\ne k$, this can be rewritten, using 
(\ref{eq:AffDistributivity}), as 
\begin{equation}
\label{eq:AffBasisExp2}
q=p_k+\sum_{k\ne a=0}^n v^a(p_a-p_k)\,,\quad
\mathrm{where}\quad v^0:=1-\sum_{a=1}^nv^a\,.
\end{equation}
\end{subequations}
This motivates writing the sums on the right hand sides 
of (\ref{eq:AffBasisExp}) in a perfectly symmetric way 
without preference of any point $p_k$:
\begin{equation}
\label{eq:AffBasisConvex Sum}
q=\sum_{a=0}^n v^ap_a,,\quad
\mathrm{where}\quad \sum_{a=0}^nv^a=1\,,
\end{equation}
where the right hand side is defined by any of the 
expressions (\ref{eq:AffBasisExp}). This defines certain 
\emph{linear combinations} of affine points, namely those whose 
coefficients add up to one. Accordingly, the \emph{affine span} 
of points $\{p_1,\cdots,p_m\}$ in affine space is defined by 
\begin{equation}
\label{eq:DefAffSpan}
\Span\{p_1,\cdots,p_n\}:=\left\{\sum_{a=1}^mv^ap_a \mid
v^a\in\field\,,\sum_{a=1}^m v^a=1\right\}\,.
\end{equation}


\end{document}